\documentclass[11pt,letterpaper]{article}
\usepackage[utf8]{inputenc}

\usepackage{typearea}
\usepackage{setspace}

\usepackage{fullpage}

\usepackage{comment} 
\usepackage{bbm}
\usepackage{framed} 
\usepackage{url} 
\usepackage{complexity}
\usepackage{booktabs}
\usepackage{amsmath,amssymb}

\usepackage{amsthm}
\usepackage{thmtools} 
\usepackage{thm-restate}
\usepackage{nicefrac}
\usepackage{calc}
\usepackage{enumitem}
\usepackage[usenames,dvipsnames]{xcolor}
\usepackage{algorithm}
\usepackage{graphicx}
\usepackage[font=footnotesize,labelfont=bf]{subcaption}
\usepackage[font=footnotesize,labelfont=bf]{caption}
\usepackage[nobreak=true]{mdframed}
\usepackage{appendix}
\usepackage[noend]{algpseudocode}
\usepackage[colorinlistoftodos]{todonotes}
\usepackage{xr}
\usepackage{array}
\usepackage{xspace}
\definecolor{ForestGreen}{rgb}{0.1333,0.5451,0.1333}
\definecolor{DarkRed}{rgb}{0.8,0,0}
\definecolor{Red}{rgb}{1,0,0}
\usepackage[linktocpage=true,colorlinks,
linkcolor=BrickRed,citecolor=blue,
bookmarks,bookmarksopen,bookmarksnumbered]{hyperref}

\usepackage[capitalise]{cleveref}
\crefrangelabelformat{enumi}{#3#1#4--#5#2#6}
\usepackage{parskip}
\usepackage{chngcntr}
\usepackage{mathtools,stackengine}
\usepackage{multirow}

\usepackage[style=alphabetic,maxbibnames=99,maxcitenames=3,doi=false,isbn=false,url=false,eprint=false,giveninits=true,useeditor=false,backend=biber]{biblatex}

\stackMath
\newcommand{\stackGeq}[1]{%
	\setbox0=\hbox{${}\mathrel{\stackon[-1pt]{\geq}{\scriptstyle\text{#1\strut}}}{}$}
	\xdef\tmpwd{\dimexpr\the\wd0\relax}
	\kern.5\tmpwd\mathclap{\box0}&\kern.5\tmpwd
}

\usepackage{nameref}
\usepackage{tikz}
\usetikzlibrary{patterns}

\allowdisplaybreaks

\DeclareMathOperator*{\argmax}{argmax}

\DeclareMathOperator*{\Exp}{\mathbb{E}}

\DeclareMathOperator*{\Var}{Var}

\renewcommand\R{\mathbb{R}}

\newcommand\eps{\varepsilon}

\newcommand\abs[1]{\lvert #1 \rvert}

\DeclarePairedDelimiterX{\expectarg}[1]{[}{]}{%
	\ifnum\currentgrouptype=16 \else\begingroup\fi
	\activatebar#1
	\ifnum\currentgrouptype=16 \else\endgroup\fi
}

\DeclarePairedDelimiterX{\nicesetarg}[1]{\{}{\}}{%
	\ifnum\currentgrouptype=16 \else\begingroup\fi
	\activatebar#1
	\ifnum\currentgrouptype=16 \else\endgroup\fi
}

\newcommand{\innermid}{\nonscript\;\delimsize\vert\nonscript\;}
\newcommand{\activatebar}{%
	\begingroup\lccode`\~=`\|
	\lowercase{\endgroup\let~}\innermid 
	\mathcode`|=\string"8000
}

\newcommand\sse{\subseteq}

\DeclareMathOperator*{\spn}{span}
\DeclareMathOperator*{\rank}{rank}

\counterwithin{equation}{section}

\usepackage{eqparbox}

\theoremstyle{plain}

\newtheorem{theorem}{Theorem}[section]
\newtheorem{definition}[theorem]{Definition}

\newtheorem{lemma}[theorem]{Lemma}
\newtheorem{fact}[theorem]{Fact}

\newtheorem{claim}[theorem]{Claim}

\newlength{\continueindent}
\usepackage{etoolbox}
\makeatletter
\newcommand*{\ALG@customparshape}{\parshape 2 \leftmargin \linewidth \dimexpr\ALG@tlm+\continueindent\relax \dimexpr\linewidth+\leftmargin-\ALG@tlm-\continueindent\relax}
\apptocmd{\ALG@beginblock}{\ALG@customparshape}{}{\errmessage{failed to patch}}
\makeatother

\makeatletter
\def\thm@space@setup{%
	\thm@preskip=\parskip \thm@postskip=0pt
}
\makeatother

\usepackage{etoolbox}
\usepackage{tikz}
\usetikzlibrary{tikzmark}
\usetikzlibrary{calc}

\errorcontextlines\maxdimen

\newcommand{\ALGtikzmarkcolor}{black}% customise this, if you want
\newcommand{\ALGtikzmarkextraindent}{4pt}% customise this, if you want
\newcommand{\ALGtikzmarkverticaloffsetstart}{-.5ex}% customise this, if you want
\newcommand{\ALGtikzmarkverticaloffsetend}{-.5ex}% customise this, if you want
\makeatletter
\newcounter{ALG@tikzmark@tempcnta}

\newcommand\ALG@tikzmark@start{%
	\global\let\ALG@tikzmark@last\ALG@tikzmark@starttext%
	\expandafter\edef\csname ALG@tikzmark@\theALG@nested\endcsname{\theALG@tikzmark@tempcnta}%
	\tikzmark{ALG@tikzmark@start@\csname ALG@tikzmark@\theALG@nested\endcsname}%
	\addtocounter{ALG@tikzmark@tempcnta}{1}%
}

\def\ALG@tikzmark@starttext{start}
\newcommand\ALG@tikzmark@end{%
	\ifx\ALG@tikzmark@last\ALG@tikzmark@starttext
	\else
	\tikzmark{ALG@tikzmark@end@\csname ALG@tikzmark@\theALG@nested\endcsname}%
	\tikz[overlay,remember picture] \draw[\ALGtikzmarkcolor] let \p{S}=($(pic cs:ALG@tikzmark@start@\csname ALG@tikzmark@\theALG@nested\endcsname)+(\ALGtikzmarkextraindent,\ALGtikzmarkverticaloffsetstart)$), \p{E}=($(pic cs:ALG@tikzmark@end@\csname ALG@tikzmark@\theALG@nested\endcsname)+(\ALGtikzmarkextraindent,\ALGtikzmarkverticaloffsetend)$) in (\x{S},\y{S})--(\x{S},\y{E});%
	\fi
	\gdef\ALG@tikzmark@last{end}%
}

\apptocmd{\ALG@beginblock}{\ALG@tikzmark@start}{}{\errmessage{failed to patch}}
\pretocmd{\ALG@endblock}{\ALG@tikzmark@end}{}{\errmessage{failed to patch}}
\makeatother
\algblock[with]{With}{EndWith}
\algblockdefx[With]{With}{EndWith}%
[1]{\textbf{with} #1 \textbf{do}}%
{}

\makeatletter
\ifthenelse{\equal{\ALG@noend}{t}}%
{\algtext*{EndWith}}
{}%
\makeatother

\newcommand{\bx}{\mathbf{x}}
\newcommand{\by}{\mathbf{y}}
\newcommand{\bp}{\mathbf{p}}
\newcommand{\nf}{\nicefrac}

\newcommand\RR{\mathbb{R}}

\newcommand{\cD}{\mathcal{D}}
\newcommand{\cF}{\mathcal{F}}
\newcommand{\cM}{\mathcal{M}}
\newcommand{\cI}{\mathcal{I}}
\newcommand{\cA}{\mathcal{A}}
\renewcommand{\cP}{\mathcal{P}}

\newcommand{\cE}{\mathcal{E}}

\newcommand{\dt}{{\text d}t}

\title{Pairwise-Independent Contention Resolution}

\author{Anupam Gupta\thanks{Department of Computer Science, New York University, New York, NY 10012. Email:
		\texttt{anupam.g@nyu.edu}. Supported in part by NSF awards CCF-1955785 and CCF-2006953.}
\and Jinqiao Hu\thanks{Peking University, Beijing. Email: cppascalinux@gmail.com.}
\and Gregory Kehne\thanks{Department of Computer Science, UT Austin, Austin, TX 78712. Email:
		\texttt{gkehne@utexas.edu}. Supported in part by NSF award CCF-2217069}
\and Roie Levin\thanks{Department of Computer Science, Rutgers University, Piscataway, NJ 08854. Email:
		\texttt{roie.levin@rutgers.edu}. Work was done while the author was a Fulbright Israel Postdoctoral Fellow.}
}

\begin{document}
\date{}
\maketitle
\begin{abstract}

  We study online contention resolution schemes (OCRSs) and prophet
  inequalities for non-product distributions. Specifically, when the
  active set is sampled according to a \emph{pairwise-independent}
  (PI) distribution, we show a $(1-o_k(1))$-selectable OCRS for
  uniform matroids of rank $k$, and $\Omega(1)$-selectable OCRSs for
  laminar, graphic, cographic, transversal, and regular matroids. These
  imply prophet inequalities with the same ratios when the set
  of values is drawn according to a PI distribution. Our results
  complement recent work of \textcite{DBLP:journals/corr/abs-2310-05240} showing that no
  $\omega(1/k)$-selectable OCRS exists in the PI setting for general
  matroids of rank $k$.

\end{abstract}

\newpage

\section{Introduction}

Consider the \emph{prophet inequality} problem: a sequence of independent positive real-valued random variables $\mathbf{X} = \langle X_1, X_2, \ldots, X_n \rangle$ are revealed one by one. Upon seeing $X_i$ the algorithm must decide whether to select or discard the index $i$; these decisions are irrevocable. The goal is to choose some subset $S$ of the indices $\{1,2,\ldots, n\}$ to maximize $\Exp[\sum_{i \in S} X_i]$, subject to the set $S$ belonging to a well-behaved family $\cI \sse 2^{[n]}$. The goal is to get a value close to $\Exp[ \max_{S \in \cI} \sum_{i \in S} X_i]$, the value that a clairvoyant ``prophet'' could obtain in expectation. This problem originally arose in optimal stopping theory, where the case of $\cI$ being the set of all singletons was considered~\cite{MR0515432}: more recently, the search for good prophet inequalities has become a cornerstone of stochastic optimization and online decision making, with the focus being on generalizing to broad classes of downward-closed sets $\cI$~\cite{MR3926869,feldman2016online,DBLP:conf/stoc/Rubinstein16}, considering additional assumptions on the order in which these random variables are revealed \cite{DBLP:journals/siamdm/EsfandiariHLM17,DBLP:conf/soda/ArsenisDK21,DBLP:journals/corr/abs-2304-04024,DBLP:conf/focs/Adamczyk018}, obtaining optimal approximation guarantees \cite{DBLP:conf/esa/LeeS18,DBLP:journals/mor/CorreaFHOV21}, and competing with nonlinear objectives \cite{DBLP:conf/soda/RubinsteinS17,DBLP:journals/corr/abs-2308-05207}.

One important and interesting direction is to reduce the requirement of independence between the random variables: \emph{what if the r.v.s are correlated}? The case of negative correlations is benign~\cite{MR1228074,MR0913569}, but general correlations present significant hurdles---even for the single-item case, it is impossible to get value much better than the trivial $\Exp[\max_i X_i]/n$ value obtained by random guessing~\cite{MR1160620}. As another example, the model with \emph{linear correlations}, where $\mathbf{X} = A \mathbf{Y}$ for some independent random variables $\mathbf{Y} \in \mathbb{R}_+^d$ and known non-negative matrix $A \in \mathbb{R}_+^{d\times n}$, also poses difficulties in the single-item case~\cite{Immorlica0W20}.

Given these impossibility results, \textcite{DBLP:conf/wine/CaragiannisGLW21} gave single-item prophet inequalities in the setting of \emph{weak correlations}: specifically, for the setting of \emph{pairwise-independent distributions}. As the name suggests, these are distributions that look like product distributions when projected on any two random variables. While pairwise-independent distributions have long been studied in algorithm design and complexity theory~\cite{DBLP:journals/fttcs/LubyW05}, these had not been studied previously in the setting of stochastic optimization. \textcite{DBLP:conf/wine/CaragiannisGLW21} gave both algorithms and some limitations arising from just having pairwise-independence. They also consider related pricing and bipartite matching problems.

We ask the question: can we extend the prophet inequalities known for richer classes of constraint families $\cI$ to the pairwise-independent case? In particular,

\begin{quote}
  \emph{What classes of matroids admit good pairwise-independent
    prophet inequalities?}
\end{quote}

\smallskip Specifically, we investigate the analogous questions for \emph{(online) contention resolution schemes (OCRS)}~\cite{feldman2016online}, another central concept in online decision making, and one closely related to prophet inequalities. In an OCRS, a random subset of a ground set is marked active. Elements are sequentially revealed to be active or inactive, and the OCRS must decide irrevocably on arrival whether to select each active element, subject to the constraint that the selected element set belongs to a constraint family $\cI$. The goal is to ensure that each element, conditioned on being active, is picked with high probability.  It is intuitive from the definitions (and formalized by \textcite{feldman2016online}) that good OCRSs imply good prophet inequalities (see also \cite{DBLP:conf/esa/LeeS18}).

\subsection{Our Results}
\label{sec:our-results}

Our first result is for the $k$-uniform matroid where the algorithm can pick up to $k$ items: we  achieve a
$(1-o_k(1))$-factor of the expected optimal value.

\begin{theorem}[Uniform Matroid PI Prophets] \label{thm:uniformprophet}
There is an algorithm in the prophet model for $k$-uniform matroids that achives expected value at least $(1-O(k^{-\nf15}))$ of the expected optimal value. 
\end{theorem}

We prove this by giving a $(1-O(k^{-\nf15}))$-\emph{selectable} online contention resolution scheme for $k$-unifrom matroids, even when the underlying generative process is only pairwise-independent. \textcite{feldman2016online} showed that selectable OCRSs immediately lead to prophet inequalities (against an almighty adversary) in the fully independent case, and we observe that their proofs translate to pairwise-independent distributions as well. Along the way we also show a $(1-O(k^{-\nf13}))$ (offline) CRS for the pairwise-independent $k$-uniform matroid case.

We then show $\Omega(1)$-selectable OCRSs for many common matroid
classes, again via pairwise-independent OCRSs.

\begin{theorem}[Other Matroids PI Prophets]
  There exist $\Omega(1)$-selectable OCRSs for laminar (\cref{sec:laminar}), graphic (\cref{sec:graphic}), cographic (\cref{sec:cographic}),
  transversal (\cref{sec:transversal}), and regular (\cref{sec:regular}) matroids. These immediate imply $\Omega(1)$
  prophet inequalities for these  matroids against almighty adversaries.
\end{theorem}

Finally, we consider the single-item case in greater detail in \Cref{sec:single-item}. For this
single-item case the current best result from \textcite{DBLP:conf/wine/CaragiannisGLW21} uses a multiple-threshold
algorithm to achieve a $(\sqrt{2}-1)$-prophet inequality; however,
this bound is worse than the $\nf{1}{2}$-prophet inequality known for
fully independent distributions. 
We show in \cref{sec:imposs-single} that
no (non-adaptive) multiple-threshold algorithm (i.e., one that
prescribes a sequence of thresholds $\tau_i$ up-front, and picks the
first index $i$ such that $X_i \geq \tau_i$) can beat
$2(\sqrt{5}-2) \approx 0.472$, suggesting that if $\nf12$ is at all
possible it will require adaptive algorithms.

\begin{theorem}[Upper Bound for Multiple Thresholds]
  Any multiple-thres\-hold algorithm for the single-item prophet
  inequality in the PI setting achieves a factor of at most $0.472$.
\end{theorem}

In \cref{sec:single-sample} we give a \emph{single-sample} version of the single-item PI
prophet inequality.

\begin{restatable}[Single Sample Prophet Inequality]{theorem}{singlesample}
    \label{thm:single_sample}
      There is an algorithm that draws a single sample from the underlying
  pairwise-independent distribution $\langle \widetilde X_1, \ldots, \widetilde X_n \rangle \sim \cD$ on $\mathbb{R}_+^n$, and
  then faced with a second sample
  $\langle X_1, \ldots, X_n \rangle \sim \cD$ (independent from $\langle \widetilde X_1, \ldots, \widetilde X_n \rangle$), picks a single item $i$ from $X_1, \ldots, X_n$
  with expected value at least
  $\Omega(1) \cdot \Exp_{\mathbf{X} \sim \cD}[\max_i X_i]$.
\end{restatable}

The results stated so far appeared in the original version of this paper \cite{DBLP:conf/ipco/GuptaHKL24}. A natural follow-up question is: what happens if one assumes $t$-wise independence instead of $2$-wise? How do the quality guarantees improve with $t$? In \cref{sec:twise}, we characterize exactly the quality of CRSs in the $t$-wise independent setting for the single item case.

\begin{restatable}[$t$-wise Independent Prophet Inequality]{theorem}{twise}
    \label{thm:twise-CRS-tight-bounds}
    For $t$-wise independent distributions, the best achievable single item CRS quality is $\phi(t) + O(1/n)$, where  
    \[\phi(t) := 1 - \sum_{k=0}^{2 \cdot \lfloor t/2\rfloor} \frac{(-1)^k}{k!}.\]
\end{restatable}

Note that this is the truncation to the first $2 \cdot \lfloor t/2\rfloor$ terms of the Taylor expansion for $1-e^{-1}$ at $0$, and converges to $1-e^{-1}$ in the fully independent case as expected \cite{DBLP:conf/focs/FeigeV06,DBLP:journals/toc/FeigeV10}.

Finally, in \cref{app:almighty-ub}, we also show that there is no $(\nicefrac{1}{4} + \eps)$-balanced OCRS against the almighty adversary (we define the almighty adversary formally in \cref{app:preliminaries}).

\subsection{Related Work}
\label{sec:related-work}

In independent and concurrent work, \textcite{DBLP:journals/corr/abs-2310-05240} also
study the pairwise-independent versions of prophet inequalities and
(online) contention resolution schemes. This work can be considered
complementary to ours: they show that for arbitrary linear matroids,
at best $O(1/r)$ factors can be achieved for
pairwise-independent versions of OCRSs, and at best
$O(1/(\log r))$ factors can be achieved for pairwise-independent
versions of matroid prophet inequalities, where $r$ is the rank of the matroid. They also obtain $\Omega(1)$-selectable OCRSs for uniform, graphical, and bounded degree transversal matroids by observing that these have the $\alpha$-partition property (see \cite{DBLP:conf/soda/BabaioffDGIT09}), reducing to the single-item setting. 
Another motivation for our work is the famous matroid secretary problem, since the latter is known to be equivalent to the existence of good OCRSs for \emph{arbitrary} distributions that admit $\Omega(1)$-balanced CRSs against a \emph{random-order} adversary \cite{DBLP:conf/innovations/Dughmi22}.

The original single-item prophet inequality for product distributions was proven by \textcite{MR0515432};  applications to computer science were given by \textcite{DBLP:conf/aaai/HajiaghayiKS07} and \textcite{DBLP:conf/stoc/ChawlaHMS10}. There is a vast literature on variants and extensions of prophet inequalities, which we cannot survey here for lack of space.
Contention resolution schemes were introduced by \textcite{DBLP:journals/siamcomp/ChekuriVZ14} in the context of constrained submodular function maximization, and these were generalized by \textcite{feldman2016online} to the online setting in order to give prophet inequalities for richer constraint families.

Limited-independence versions of prophet inequalities were studied from the early days e.g. by \textcite{MR1160620,MR1228074,MR1160620}. Many stochastic optimization problems have been studied recently in correlation-robust settings, e.g., by \textcite{DBLP:conf/esa/BateniDHS15, DBLP:conf/approx/0001GMT23,Immorlica0W20,DBLP:journals/corr/abs-2310-07809}; pairwise-independent prophet inequalities were introduced by \textcite{DBLP:conf/wine/CaragiannisGLW21}.

There is a growing body of work on single-sample prophet inequalities in the i.i.d. setting \cite{azar14prophet,azar19prior,rubinstein20optimal,caramanis21single,DBLP:journals/corr/abs-2304-02063}. We are the first to study these for pairwise-independent distributions.

\section{Preliminaries}
\label{app:preliminaries}

We assume the reader is familiar with the basics of matroid theory, and refer to \cite{Schrijver-book} for definitions. We will call elements of the ground set \emph{parallel} if they form a circuit of size $2$. The matroid polytope is defined to be  $\cP_M := \{\bx \in \RR^{E} \mid \bx\geq \mathbf 0,\: \bx(S) \leq \rank(S) \:\: \forall S \sse E\}.$

We focus on pairwise-independent versions of \emph{contention resolution schemes} (CRSs), both in the offline and online settings. For polytope $\cP$ and scalar $b \in \RR$, define $b\cP := \{b \bx \mid \bx \in \cP\}$.

\begin{definition}[PI-consistency]
  \label{def:pi-cons}
  Given a polytope $\cP \sse [0,1]^n$ and a vector $\bx = (x_1, \ldots, x_n) \in \cP$, a probability distribution $\cD$ over subsets of $[n]$ is \emph{PI-consistent} with $\bx$ if
  \begin{enumerate}[nosep]
    \item $\Pr_{R \sim \cD}[i \in R] = x_i$, for all $i\in[n]$, and
    \item $\Pr_{R \sim \cD}[\{i,j\} \sse R] = x_i \cdot x_j$ for all $i,j\in[n],i\ne j$.
  \end{enumerate}
\end{definition}

\begin{claim}
  \label{clm:scale}
  Consider the distribution $\cD'$ obtained by drawing $R \sim \cD$
  and then subsampling each element \emph{independently} with
  probability $p \in [0,1]$. If $\cD$ is PI-consistent for $\bx$ then
  $\cD'$ is PI-consistent for $b\bx \in b\cP$.
\end{claim}

\begin{definition}[$(b,c)$-balanced PI-CRS]
  \label{def:picrs}
  Let $b,c \in [0,1]$. A \emph{$(b,c)$-balanced pairwise-independent contention resolution scheme} (or \emph{$(b,c)$-balanced PI-CRS} for short) for a convex polytope $\cP$ is a (randomized) algorithm $\pi:2^{[n]} \to 2^{[n]}$ that given any distribution $\cD$ over subsets of $[n]$ and any set $R\sse[n]$, returns a random set $\pi(R)$ satisfying the following properties:
  \begin{enumerate}[nosep]
    \item (subset) $\pi(R) \sse R$,
    \item (feasibility) $\chi_{\pi(R)} \in \cP$, and
    \item ($(b,c)$-balanced) Let $\bx$ be any vector in $b\cP$, and $\cD$ be any distribution PI-consistent with $\bx$. Let $R$ be sampled according to $\cD$, then for any $i \in [n]$,
    $\Pr[ i \in \pi(R) \mid i \in R] \geq c$, where the probability is over
    both the distribution $\cD$ and any internal randomization of $\pi$. 
  \end{enumerate}
  We use \emph{$c$-balanced PI-CRS} as an abbreviation for $(1,c)$-balanced PI-CRS.
\end{definition}

\begin{claim}\label{clm:scaleCRS}
    Given a $(b,c)$-balanced PI-CRS $\pi$ for a polytope $\cP$, we can construct a $bc$-balanced PI-CRS $\pi'$ for $\cP$.
\end{claim}
\begin{proof}
    The PI-CRS $\pi'$ works as follows: on receiving $\bx\in\cP$ and $R\sim \cD$ where $\cD$ is PI-consistent with $\bx$, flips a coin for each element $e\in R$ and discards $e$ with probability $1-b$. Let $R'\subseteq R$ denote the resulting set and $\cD'$ the underlying distribution. The PI-CRS $\pi'$ then returns $\pi(R')$.

    By \Cref{clm:scale}, $\cD'$ is PI-consistent with $b\bx$. Let $I$ be the set that $\pi(R')$ returns. Then $I\sse R'\sse R$, and $\chi_I\in\cP$. Also $\Pr[i\in I\mid i\in R]=b\Pr[i\in I\mid i\in R']\ge bc$. Therefore we conclude that $\pi'$ is indeed a $bc$-balanced PI-CRS for $\cP$.
\end{proof}

An \emph{online} contention resolution scheme is one where the $n$ items are presented to the algorithm one-by-one in some order chosen by the adversary; the algorithm has to decide whether to accept an arriving item into the resulting independent set, or to irrevocably reject it~\cite{feldman2016online}. The precise definition of PI-OCRS depends on the adversarial model. We define these precisely in \Cref{app:preliminaries}, and give just the high-level idea here.

\subsection{The Almighty Adversary}
The almighty adversary knows everything about the game even before the game starts. More specifically, first the adversary sees the realization of $R \sim \cD$, together with the outcome of all the random coins the algorithm uses. It then presents the items in $R$ to the algorithm in the worst possible order. (Notice that since all the random coins of $\pi$ is already tossed and known to the adversary, $\pi$ is a deterministic algorithm from the perspective of the almighty adversary; therefore with enough computational power the almighty adversary can compute the worst order of items for $\pi$.) To describe a kind of PI-OCRS effective against the almighty adversary, we adopt ideas from \cite{feldman2016online}.

\begin{definition}[Greedy PI-OCRS]
    Let $\cP\sse[0,1]^n$ be some convex polytope, and $\cD$ be some distribution over subsets of $[n]$. We call a (randomized) algorithm $\pi: 2^{[n]} \to 2^{[n]}$ a \emph{greedy PI-OCRS} if:
    \begin{enumerate}
        \item Before the items arrive, $\pi$ has access to $\cD$, and (randomly) defines some down-closed set family $\cF\sse\cI$.
        \item Let $I$ be the set $\pi$ selects, initially empty. When an item $e$ arrives, it is added to $I$ if $I\cup\{e\}\in\cF$.
        \item Finally, $\pi$ returns $I$.
    \end{enumerate}
    We say $\cF$ is the \emph{feasible set family} defined by $\pi$, and for simplicity, we say the step 2 and 3 of this algorithm is a greedy PI-OCRS defined by $\cF$.
\end{definition}
Here $\pi$ need not list $\cF$ explicitly. In fact, the ability to judge whether a given set $S$ is in $\cF$ is enough. We now define the notion of selectable PI-OCRS, which is powerful even against the almighty adversary:

\begin{definition}[$(b,c)$-selectable PI-OCRS]
    Let $\cP\sse[0,1]^n$ be some convex polytope. We call a (randomized) algorithm $\pi: 2^{[n]}\to2^{[n]}$ a \emph{$(b,c)$-selectable PI-OCRS} if it satisfies the following:
    \begin{enumerate}
        \item Algorithm $\pi$ is a \emph{greedy} PI-OCRS.
        \item For any $\bx\in b\cP$, any distribution $\cD$ over subsets of $[n]$ that is PI-consistent with $\bx$ and any $i\in[n]$, let $\cF$ be the feasible set family defined by $\pi$. Let $R$ be sampled according to $\cD$, then 
        \begin{equation}
            \Pr_{R\sim\cD}[I\cup\{i\}\in\cF\quad\forall I\sse R, I\in\cF\mid i\in R]\ge c.\label{eqn:selectable}
        \end{equation}
        Here the probability is over $R$ and internal randomness of $\pi$ in defining $\cF$.
    \end{enumerate}
    We use \emph{$c$-selectable PI-OCRS} as an abbreviation for $(1,c)$-selectable PI-OCRS.
\end{definition}
Notice the definition here is slightly different from \cite{feldman2016online}, as we need to condition on the event $i\in R$. This is due to our limited independence over events $i\in R$: For the mutually independent case, one can prove that $\Pr_{R\sim\cD}[I\cup\{i\}\in\cF\quad \forall I\sse R, I\in\cF\mid i\in R]=\Pr_{R\sim\cD}[I\cup\{i\}\in\cF\quad \forall I\sse R, I\in\cF]$. But this may not hold in the pairwise-independent case.

We now argue that any $(b,c)$-selectable PI-OCRS is powerful even
against an almighty adversary.
\begin{claim}
    Let $\cP\sse[0,1]^n$ be a convex polytope, and $\pi$ a $(b,c)$-selectable PI-OCRS. Let $\bx$ be any vector in $b\cP$, and $\cD$ be any distribution PI-consistent with $\bx$. Let $R$ be sampled according to $\cD$. Then for any item $i\in[n]$, $\Pr[i\in\pi(R)\mid i\in R]\ge c$, even when the order of the items is decided by the almighty adversary.
\end{claim}
\begin{proof}
  The constraint \eqref{eqn:selectable} guarantees that no matter what set $I$ we choose before $i$ arrives, the event $I\cup\{i\}\in\cF$ holds (and therefore $i$ is selected when it is in $R$) with probability at least $c$. Since this argument holds regardless of the items appearing before $i$, we conclude that $\Pr[i\in\pi(R)\mid i\in R]\ge c$.
\end{proof}

Using the same idea of subsampling each item independently as \Cref{clm:scale}, we can also obtain a $bc$-selectable PI-OCRS from a $(b,c)$-selectable PI-OCRS.
\begin{claim}\label{clm:scaleOCRS}
    Given a $(b,c)$-selectable PI-OCRS $\pi$ for some polytope $\cP$, we can construct a $bc$-selectable PI-OCRS $\pi'$ for $\cP$.
\end{claim}
\begin{proof}
    Let $\cD$ be any distribution PI-consistent with some vector $\bx\in\cP_M$. First, $\pi'$ defines a random set $S$, where for each item $i$, $\pi'$ flips a random coin to include $i$ in $S$ with probability $b$. Let $R$ be sampled according to $\cD$, and let $\cD'$ denote the distribution of $R\cap S$, then one can verify that $\cD'$ is PI-consistent with $b\bx$. Let $\cF$ denote the feasible set family defined by $\pi$ for $\cD'$. Then $\pi'$ defines $\cF'$ as
    \[
        \cF':=\{I\in\cF\mid I\sse S\}
    \]
    Here $\cF'\sse\cF\sse\cI$, and it remains to bound the selectability of the greedy PI-OCRS defined by $\cF'$. Observe that for any $i\in S$, the following two propositions are equivalent:
    \begin{enumerate}
        \item $I\cup\{i\}\in\cF$ holds for any $I\in\cF,I\sse R\cap S$.
        \item $I\cup\{i\}\in\cF'$ holds for any $I\in\cF',I\sse R$.
    \end{enumerate}
    Then for any item $i$, we have
    \begin{align*}
        &\Pr[I\cup \{i\}\in\cF'\quad\forall I\in \cF',I\sse R\mid i\in R]\\
        &=\Pr[i\in S\mid i\in R]\Pr[I\cup \{i\}\in\cF'\quad\forall I\in \cF',I\sse R\mid i\in R\cap S]\\
        &=b\Pr[I\cup\{i\}\in\cF\quad\forall I\in\cF,I\sse R\cap S\mid i\in R\cap S]\\
        &\ge bc.\tag{$R\cap S\sim\cD'$, $(b,c)$-selectability of $\pi$}
    \end{align*}
    This completes the proof.
\end{proof}

\subsection{The Offline Adversary}
A much weaker type of adversary is the offline adversary, who knows nothing about the outcome of the random events, and has to decide the order of the items before they start to arrive. In other words, the adversary presents the worst, fixed order to the algorithm before $R \sim \cD$ is drawn or the randomness of the algorithm is determined. Here we use \emph{$(b,c)$-balanced} to denote PI-OCRS effective against offline adversary, compared to the \emph{$(b,c)$-selectable} against almighty adversary.

\begin{definition}[$(b,c)$-balanced PI-OCRS]\label{def:balancedOCRS}
    Let $\cP\sse[0,1]^n$ be some convex polytope. We call a (randomized) algorithm $\pi:2^{[n]} \to 2^{[n]}$ a $(b,c)$-balanced PI-OCRS for $\cP$, if given any distribution $\cD$ over subsets of $[n]$ and any set $R\in E$, returns a set $\pi(R)$ satisfying the following properties:
    \begin{enumerate}
        \item (subset) $\pi(R) \sse R$,
        \item (feasibility) $\chi_{\pi(R)} \in \cP$, and
        \item ($(b,c)$-balanced) Let $\bx$ be any vector in $b\cP$, and $\cD$ be any distribution PI-consistent with $\bx$, and let $R$ be sampled according to $\cD$. Suppose the adversary has to choose the order of the items before $R$ is sampled or the random coins of $\pi$ are tossed, and when each item $i$ is presented to $\pi$ and is in $R$, $\pi$ can make an irrevocable choice to select it or discard it. Then for any $i$,
        \[
            \Pr[i\in \pi(R)\mid i\in R]\ge c.
        \] 
    \end{enumerate}
    We use $c$-balanced PI-OCRS as an abbreviation for $(1,c)$-balanced PI-OCRS.
\end{definition}

Using the same idea of subsampling each item independently as \Cref{clm:scaleCRS}, we can also obtain a $bc$-balanced PI-OCRS from a $(b,c)$-balanced PI-OCRS. The proof is similar to \Cref{clm:scaleCRS} and therefore omitted.
\begin{claim}
    Given a $(b,c)$-balanced PI-OCRS $\pi$ for some polytope $\cP$, we can construct a $bc$-balanced PI-OCRS $\pi'$ for $\cP$.
\end{claim}

The property of \emph{balancedness} is weaker than that of \emph{selectability}: $(b,c)$-selectable PI-OCRSs are also $(b,c)$-balanced, but the converse might not be true. However, a \emph{balanced} PI-OCRS is enough to solve the prophet inequality problem against offline adversary, as we will see in the next section.

\subsection{Prophet Inequalities and OCRSs}

Feldman et al.~\cite{feldman2016online} showed connections between OCRSs and prophet inequalities; here we extend this relationship to the pairwise-independent setting.

\begin{definition}[Pairwise-Independent Matroid Prophet Game]
    Let $M=(E,\cI)$ be a matroid. For simplicity, we identify $E$ as $[n]$. Let $W_1,\ldots,W_n$ be $n$ pairwise-independent non-negative random variables denoting the weight of the items, and let $\cD_W$ be their joint distribution, which is known to the gambler. Now the $n$ items are presented to the gambler one by one in some \emph{fixed} order \emph{unknown} to the gambler, and when each item arrives, the gambler has to make an irrevocable choice to select the item or not. The items selected must form an independent set of $M$. The gambler wants to maximize the sum of weights of the items he selects. We define the offline optimum $\mathsf{OPT}:=\mathbb E[\max_{I\in\cI}\sum_{i\in I}W_i]$, and let $\mathsf{ALG}$ denote the expected reward of the algorithm the gambler uses. We call an algorithm an $\alpha$-approximation if $\mathsf{ALG}/\mathsf{OPT}\ge\alpha$.
\end{definition}

We claim that a good PI-OCRS gives a good approximation for the pairwise-independent matroid prophet game:

\begin{claim}
    Let $M=(E,\cI)$ be a matroid. If there is a $c$-balanced PI-OCRS for $\cP_M$, then there is a $c$-competitive algorithm for the pairwise-independent matroid prophet game over $M$.
\end{claim}
\begin{proof}
    For simplicity, we assume all the $W_i$ are distributed continuously. (The extension to discrete distributions is standard, see e.g. \cite[page~2]{DBLP:conf/innovations/RubinsteinWW20}.) Let $I_{\max}$ be an arbitrary optimal set, i.e., $I_{\max}:=\argmax_{I\in\cI}\sum_{i\in I}W_i$. Let $p_i :=\Pr[i\in I_{\max}]$ denote the probability that element $i$ is in the optimal set. Then we claim that $\mathbf p$ is in $\cP$. This is because for any $S\sse E$,
    \[
        \sum_{i\in S} p_i=\sum_{i\in S}\Pr[i\in I_{\max}]=\mathbb E[|I_{\max}\cap S|]\le\rank(S).
    \]
    After we obtain $\mathbf p$, we set a threshold $\tau_i$ for each $i\in E$ such that $\Pr[W_i\ge\tau_i]=p_i$. We then have the following lemma regarding the expected optimum value.

    \begin{lemma}\label{lemma:optprob}
        Let $M=(E,\cI)$ be a matroid, and $W_i,p_i,\tau_i$ be defined as above. Then
        $$
            \mathsf{OPT}=\mathbb E\bigg[\max_{I\in\cI}\sum_{i\in I}W_i\bigg]\le\sum_{i\in E}p_i\cdot\mathbb E[W_i\mid W_i\ge\tau_i].
        $$
    \end{lemma}
    \begin{proof}
        First notice that $\mathsf{OPT}=\sum_{i\in E}p_i\cdot\mathbb E[W_i\mid i\in I_{\max}]$. Since $p_i=\Pr[i\in I_{\max}]=\Pr[W_i\ge\tau_i]$, we have
        \begin{align*}
            p_i\: \mathbb E[W_i|i\in I_{\max}]&=p_i\int_{t\ge0}\Pr[W_i\ge t\mid i\in I_{\max}]\, \dt =\int_{t\ge0}\Pr[W_i\ge t, i\in I_{\max}]\, \dt\\
            &\le\int_{t\ge0}\Pr[W_i\ge t,W_i\ge\tau_i]\, \dt =p_i\int_{t\ge0}\Pr[W_i\ge t\mid W_i\ge\tau_i]\, \dt\\
            &=p_i\: \mathbb E[W_i\mid W_i\ge\tau_i].
        \end{align*}
        Summing over $i$ finishes the proof.
    \end{proof}
    Now we describe the procedure for translating PI-OCRSs into approximations for prophet games. Let $\pi$ be a $c$-balanced PI-OCRS for $M$. Let the random set $R$ denote the items with weights exceeding $\tau_i$, i.e., $R:=\{i\mid W_i\ge\tau_i\}$. Let $\cD$ denote the distribution of $R$, and $\cD$ is known to $\pi$ before the items begin to arrive. When item $i$ arrives, we add $i$ to $R$ and send it to $\pi$ if $W_i\ge\tau_i$. Then we select $i$ if and only if $\pi$ selects $i$.

    First, by our argument before, $\mathbf p$ is in $\cP_M$. Then notice that since the $W_i$ are pairwise-independent, the events $\{i\in R\}$, or $\{W_i\ge\tau_i\}$, are also pairwise-independent. Therefore $R$ is PI-consistent with $\mathbf p$. Then the definition of $c$-balanced PI-OCRS guarantees that for each $i$, $\Pr[i\in \pi(R)\mid i\in R]\ge c$. Since whether $i$ is selected by $\pi$ when it is in $R$ is independent of $W_i$, we have $\mathbb E[W_i\mid i\in \pi(R)]=\mathbb E[W_i\mid i\in R]=\mathbb E[W_i\mid W_i\ge\tau_i]$. Then we can bound the expected reward of our algorithm by
    \begin{align*}
        \mathsf{ALG}&=\sum_{i\in E}\Pr[i\in \pi(R)]\cdot\mathbb E[W_i\mid i\in \pi(R)] =\sum_{i\in E}\Pr[i\in \pi(R)]\cdot\mathbb E[W_i\mid W_i\ge\tau_i]\\
        &\ge\sum_{i\in E}c\Pr[i\in R]\cdot\mathbb E[W_i\mid W_i\ge\tau_i] \tag{\Cref{def:balancedOCRS}}\\
        &=\sum_{i\in E}c\,p_i\cdot\mathbb E[W_i\mid W_i\ge\tau_i] \ge c\,\mathsf{OPT}. \tag{\Cref{lemma:optprob}}
    \end{align*}
    Therefore we conclude that our algorithm gives a $c$-approximation for the pairwise-independent matroid prophet game.
\end{proof}

\section{Uniform Matroids}
\label{sec:uniform}

Recall that the independent sets of a \emph{uniform matroid} $M = (E, \cI)$
of rank $k$ are all subsets of $E$ of size at most $k$; hence our goal is to
pick some set of size at most $k$. Identifying $E$ with $[n]$, the
corresponding matroid polytope is $\cP_M := \{ \bx \in [0,1]^n \mid \sum_{i=1}^n x_i \leq k \}.$
Our main results for uniform matroids are the following, which imply~\Cref{thm:uniformprophet}.

\begin{restatable}[Uniform Matroids]{theorem}{Uniform}
  \label{thm:uniform}
  For uniform matroids of rank $k$, there is
  \begin{enumerate}[nosep,label=(\roman*)]
  \item a $(1 - O(k^{-\nicefrac13}))$-balanced PI-CRS, and 
  \item a $(1 - O(k^{-\nicefrac15}))$-selectable PI-OCRS.
  \end{enumerate}
\end{restatable}

A simple greedy PI-OCRS follows by choosing the feasible set family $\cF=\cI$, i.e. selecting $R$ as the resulting set if $|R|\le k$. However, conditioning on $i\in R$, pairwise independence only guarantees the
marginals of the events $j\in R$ (and they might have arbitrary correlation), so we can only use Markov's inequality to bound $\Pr[|R|\le k\mid i\in R]$. Therefore this analysis only gives a $(b,1-b)$-selectable PI-OCRS for $k$-uniform matroid. (For details
see~\Cref{sec:app-uniform-simpleOCRS}.)

Hence instead of
conditioning on some $i\in R$ and using Markov's inequality, we consider all of the items together, and use Chebyshev's inequality to give a bound for $\Pr[|R|\ge k, i\in R]$. The following lemma is key for both our PI-CRS and PI-OCRS.

\begin{restatable}{lemma}{averagingtwo}
    \label{clm:averaging2}
    Let $M=(E,\cI)$ be a $k$-uniform matroid, where $E$ is identified as $[n]$. Given $\bx \in (1-\delta)\cP_M$ and a distribution $\cD$ over subsets of $E$ that is PI-consistent with $\bx$, let $R\sse E$ be the random set sampled according to $\cD$. Then 
    \begin{gather*}
        \sum_{i=1}^n\Pr[\abs{R}\ge k,i\in R]\le\frac{1-\delta^2}{\delta^2}.
    \end{gather*}
\end{restatable}
\begin{proof}
    The expression on the left can be written as
    \begin{align*}
        \sum_{i=1}^n\Pr[\abs{R}\ge k, i\in R] &= \sum_{i=1}^n \sum_{t=k}^n \sum_{\substack{S: \\ \abs{S} = t}} \mathbbm{1}[i\in S] \Pr[R=S]
            =  \sum_{t=k}^n \sum_{\substack{S: \\ \abs{S} = t}} \Pr[R=S] |S| \\
                                             &=\sum_{t=k}^n t\cdot\Pr[\abs{R}=t] 
                                             =k\Pr[\abs{R}\ge k]+\sum_{t=k+1}^n\Pr[\abs{R}\ge t].
      \end{align*}
      We now bound the two parts separately using Chebyshev's inequality.
      Let $X_i:=\mathbbm{1}[i\in R]$ be the indicator
      for
      $i$ being active, and let $X = \sum_{i\in E} X_i$.
      Since the $X_i$ are pairwise independent,
      $\Var[X]=\sum_i \Var[X_i] \le \sum_i \Exp[X_i^2] = \sum_i \Exp[X_i]
      = \Exp[X]$. For the first part, we have
    \begin{align}
      k \cdot \Pr[\abs{R}\ge k]&=k \cdot \Pr[X\ge k] \le k\cdot \frac{\Var[X]}{(k-\mathbb{E}[X])^2} \tag{Chebyshev's ineq.}\\
                           &\le k\cdot\frac{1-\delta}{\delta^2k} =\frac{1-\delta}{\delta^2}. \label{eq:firstcheby}
    \end{align}
    For the second part,
    \begin{align}
      \sum_{t=k+1}^n\Pr[\abs{R}\ge t]&=\sum_{t=k+1}^n \Pr[X\ge t] \le \sum_{t=k+1}^n\frac{\Var[X]}{(t-\mathbb E[X])^2} \tag{Chebyshev's ineq.}\\
                                 &\le \sum_{t=k+1}^n\frac{(1-\delta) k}{(t-(1-\delta)k)^2} 
                                 \le
                                   (1-\delta)k\cdot\sum_{t\ge1}\frac1{(\delta
                                   k+t)^2}  \notag\\
                                 &\le (1-\delta)k\cdot\frac1{\delta k} =\frac{1-\delta}{\delta}, \label{eq:secondcheby}
    \end{align}
    where we used the inequality
    \begin{align*}
      \sum_{j\ge1}\frac1{(x+j)^2} \le\sum_{j\ge1}\frac1{(x+j-1)(x+j)}
                                 =\sum_{j\ge1}\left(\frac1{x+j-1}-\frac1{x+j}\right)
                                 =\frac1x.
    \end{align*}
    Summing up the \eqref{eq:firstcheby} and \eqref{eq:secondcheby} finishes the proof.
\end{proof}

Using this lemma, we can bound $\min_i\Pr[\abs{R}\ge k\mid i\in R]$ and obtain a $(1-O(k^{-1/3}))$-balanced PI-CRS (in the same way that \cite[Lemma 4.13]{DBLP:journals/siamcomp/ChekuriVZ14} implies a $(b,1-b)$-CRS in the i.i.d. setting). The details are deferred to \Cref{sec:uniform-CRS}.

\subsection{A $(1-O(k^{-\nicefrac15}))$-Selectable PI-OCRS for Uniform Matroids} \label{sec:uniformOCRS}

Our PI-CRS has to consider the elements in a specific
order, and therefore it does not work in the online setting where the items come in adversarial order. The key idea for our PI-OCRS is to separate ``good'' items and ``bad'' items, and control each part separately. Let us assume $R$ is sampled according to some distribution $\cD$ PI-consistent with $\bx$, and that $\bx$ is on a face of $(1-\eps)\cP_M$, i.e.
\begin{gather}
  \sum_{i=1}^n \Pr[i\in R]=(1-\eps)k.
\end{gather}
We will choose the value of $\eps$ later. For some other constants $r,b\in(0,1)$ define an item $i$ to be \emph{good} if $\Pr[\abs{R}>\lfloor(1-r\eps)k\rfloor \mid i\in R]\le b$. Let $E_g$ denote the set of good items, and $E_b:= E\setminus E_g$ the remaining \emph{bad} items. Our algorithm keeps two buckets, one for the good items and one for the bad, such that
\begin{enumerate}[nosep,label=(\roman*)]
    \item the good bucket has a capacity of $\lfloor(1-r\eps)k\rfloor$, and
    \item the bad bucket has a capacity of $\lceil r\eps k\rceil$.
\end{enumerate}
When an item arrives, we put it into the corresponding bucket as long
as that bucket is not yet full. Finally, we take the union of the items
in the two buckets as the output of our OCRS. This algorithm is indeed a greedy PI-OCRS with the feasible set family
$$
    \cF=\{I\in\cI\mid |I\cap E_g|\le\lfloor(1-r\eps)k\rfloor,|I\cap E_b|\le\lceil r\eps k\rceil\}.
$$
We show that for any item $i$, $\Pr[I\cup\{i\}\in\cF\quad\forall I\in\cF,I\sse R\mid i\in R]\ge1-o(1)$. First, for a good item $i$, by definition
\begin{align}
    &\Pr[I\cup\{i\}\in\cF \ \ \forall I\in\cF, \ I\sse R\mid i\in R] \notag \\
    &=1-\Pr[|R\cap E_g|>\lfloor(1-r\eps)k\rfloor \mid i\in R] \notag \\
    &\ge1-\Pr[\abs{R}>\lfloor(1-r\eps)k\rfloor \mid i\in R]
    \ge1-b, \notag
    \intertext{since $i$ is good. Next, for a bad item $i$,
we can use Markov's inequality conditioning on $i\in R$:}
    &\Pr[I\cup\{i\}\in\cF \ \ \forall I\in\cF,I\sse R\mid i\in R]\notag \\
    &=1-\Pr[|R\cap E_b| >\lceil r\eps k\rceil \mid i\in R] \notag \\
    &\ge1-\frac{\sum_{j\in E_b}\Pr[j\in R\mid i\in R]}{r\eps k} 
    =1-\frac{\sum_{j\in E_b}\Pr[j\in R]}{r\eps k}, \label{eq:3}
\end{align}
where we use Markov's inequality, and 
the last step uses pairwise independence of events $i\in R$. We now need to bound $\sum_{j\in E_b}\Pr[j\in R]$. If we define $\eps'$ as
$1-\eps'=\frac{1-\eps}{1-r\eps}$, then we have 
\begin{align*}
\sum_{j \in E_b} \Pr[j \in R] &= \sum_{j \in E_b} \frac{\Pr[\abs{R}\ge\lfloor(1-r\eps)k\rfloor, j\in R]}{\Pr[\abs{R}\ge\lfloor(1-r\eps)k\rfloor \mid j\in R]} \\
&\leq \sum_{j \in E_b} \frac{\Pr[\abs{R}\ge\lfloor(1-r\eps)k\rfloor, j\in R]}{b} \tag{since $j$ is bad}\\
&\stackrel{(\star)}{\leq} \frac{(1-(\eps')^2)/(\eps')^2}{b} 
 \leq \frac{1}{(1-r)^2\eps^2 b},
\end{align*}
where $(\star)$ uses \Cref{clm:averaging2}.
Substituting back into~(\ref{eq:3}),
\[\Pr[I\cup\{i\}\in\cF\quad\forall I\in\cF,I\sse R\mid i\in R]\ge1-((1-r)^2r\eps^3bk)^{-1}.\]
To balance the good and bad items, we set $b =
((1-r)^2r\eps^3bk)^{-1} = ((1-r)^2r\eps^3k)^{-\nicefrac12}$.
If we set $r=\nicefrac13$,
then we have an $(1-\eps,1-(\frac{4}{27}\eps^3k)^{-\nicefrac12})$-selectable PI-OCRS. Finally, if we set $\eps=k^{-\nicefrac15}$, then by \Cref{clm:scaleOCRS} we have a $(1-O(k^{-\nicefrac15}))$-selectable PI-OCRS.

\section{Laminar Matroids}
\label{sec:laminar}

In this section we give an $\Omega(1)$-selectable PI-OCRS for laminar
matroids. A laminar matroid is defined by a laminar family
$\cA$ of subsets of $E$, and a capacity function
$c:\cA \rightarrow \mathbb{Z}$; a set $S \sse E$ is independent if
$|S \cap A| \leq c(A)$ for all $A \in \cA$.

The outline of the algorithm is as follows: we construct a new capacity function $c'$ by rounding down $c(A)$ to powers of two; satisfying these more stringent constraints loses only a factor of two. Then we run greedy PI-OCRSs for uniform matroids from
\Cref{sec:uniformOCRS} \emph{independently} for each capacity constraint $c'(A)$, $A \in \cA$. Finally, we output the intersection of these feasible sets. For our analysis, we apply a union bound on probability of an item being discarded by some greedy PI-OCRS; this is a geometric series by our choice of $c'$.

As the first step, we define $c'(A)$ to be the largest power of $2$
smaller than $c(A)$, for each $A \in \cA$. (For simplicity we assume
that $E \in \cA$.) Moreover, if sets $A, B \in \cA$ with $A \sse B$
and $c'(A) \geq c'(B)$, then we can discard $A$ from the collection.  
In conclusion, the final constraints satisfy:
\begin{enumerate}[nosep]
    \item The new laminar family is $\cA'\sse\cA$.
    \item For any $A\in\cA'$, $c'(A)$ is power of 2, and $c(A)/2<c'(A)\le c(A)$.
    \item (Strict Monotonicity) For any $A,B\in\cA'$ with $A\subsetneq B$, we have $c'(A)<c'(B)$.
\end{enumerate}
Let $M'$ denote the laminar matroid defined by the new set of
constraints. We can check that any $c$-selectable PI-OCRS for $M'$ is
a $(\nf12,c)$-selectable PI-OCRS for $M$.  Hence, it suffices to give
a $\Omega(1)$-selectable PI-OCRS for $M'$.

Now we run greedy OCRSs for uniform matroids to get a
$(\nf{1}{25}, \nf{1}{2.661})$-selectable PI-OCRS: for a set $A$ with
capacity $c'(A)$, from \Cref{sec:uniform} we have both a $(1-b,b)$-selectable PI-OCRS and a
$(1-b,1-(\frac{4}{27}b^3c'(A))^{-1/2})$-selectable PI-OCRS: the former
is better for small capacities, whereas the latter is better for
larger capacities. Setting a threshold of $t=13$ and choosing $b = \nf{24}{25}$, we use the former when $c'(A)< 2^t$, else we use the latter. Now a union bound over the the various sets containing an element gives us the result: the crucial fact is that we get a
contribution of $t(1-b)$ from the first smallest scales and a geometric sum giving $O(2^{-t/2}b^{-3/2})$ from the larger ones. The details appear in \Cref{sec:app-laminar}.

\section{Graphic Matroids}
\label{sec:graphic}

Recall that graphic matroids correspond to forests (acyclic subgraphs)
of a given (multi)graph. For these matroids we show the following.
\begin{theorem}
    \label{thm:graphic}
  For $b\in(0,\nicefrac12)$, there is a $(b,1-2b)$-selectable PI-OCRS scheme for 
  graphic matroids. 
\end{theorem}

\begin{proof}

Let $M=(E,\cI)$ be a graphic matroid defined on (multi)graph $G=(V,E)$. Let $\cD$ be any distribution over subsets of $E$ that is PI-consistent with some $\bx\in b\cP_M$, and $R$ sampled according to $\cD$.
We follow the construction of OCRS of
\textcite{feldman2016online}. Our goal
is to construct a chain of sets: $\varnothing=E_l\subsetneq E_{l-1}\subsetneq\cdots\subsetneq E_0=E$
where for any $i\in \{0,\ldots,l-1\}$ and any $e\in E_i\setminus E_{i+1}$,
\begin{equation}
    \Pr\Big[e\in\spn_{M/E_{i+1}}((R\cap (E_i\setminus E_{i+1}))\setminus\{e\})\ \Big|\ e\in R\Big]\le 2b.\label{ineq1}
\end{equation}
If we have this chain, then we can define the feasible set for our greedy PI-OCRS as $\cF=\{I\sse E\mid\forall i,I\cap(E_i\setminus E_{i+1})\in\cI(M/E_{i+1})\}$.
By definition of contraction, $\cF\sse\cI(M)$. It remains to check selectability. For an edge $e$ in $E_i\setminus E_{i+1}$, we have
\begin{align*}
    &\Pr[I\cup\{e\}\in\cF\quad\forall I\in\cF,I\sse R\mid e\in R]\\
    &=\Pr\Big[e\notin\spn_{M/E_{i+1}}((R\cap (E_i\setminus E_{i+1}))\setminus\{e\})\ \Big|\ e\in R\Big] \ge1-2b.\tag{\Cref{ineq1}}
\end{align*}
Therefore this is a $(b,1-2b)$-selectable PI-OCRS.
All is left is to show how to construct such a chain. We use the following recursive procedure:
\begin{enumerate}[nosep]
    \item Initialize $E_0=E,i=0$.
    \item Set $S=\varnothing$.
    \item While there exists $e\in E_i\setminus S$ such that $\Pr[e\in\spn_{M/S}((R\cap(E_i\setminus S))\setminus\{e\})\mid e\in R]>2b$, add $e$ into $S$.
    \item $i \leftarrow i+1$, set $E_i=S$.
    \item If $E_i\ne\varnothing$, goto step 2; otherwise set $l=i$ and terminate.
\end{enumerate}
Inequality \eqref{ineq1} is automatically satisfied by
this procedure. It remains to show that the process
always terminates, i.e. that step 3 always leaves at least one element
unidentified, and hence $E_{i+1}\subsetneq E_i$. Let $\mathcal E(u)$
denote the edges incident to the vertex $u$. We start with the following claim.

\begin{claim}\label{clm:graphic-good-vertex}
    If $u_0\in V$ satisfies $\sum_{e\in\mathcal E(u_0)}x_e\le 2b$, then in the above procedure generating $E_1$ from $E$, we have that 
    for all $ e\in\mathcal E(u_0)$, $e\notin S$.
\end{claim}
\begin{proof}
    We prove our claim using induction. For any edge $e\in\cE(u_0)\cap R$, $e\in\spn(R\setminus\{e\})$ implies the existence of a circuit $C\sse R$ containing $e$. By the definition of circuits, $C$ must contain some edge $e'\in\cE(u_0)\setminus\{e\}$. By the pair-wise independence of events $e\in R$, we have
    \begin{align*}
        \Pr[e\in\spn(R\setminus\{e\})\mid e\in R]
        &\le\Pr[\exists e'\in\mathcal E(u_0)\setminus\{e\},e'\in R\mid e\in R]\\
        &\le \sum_{e\in\mathcal E(u_0)}x_e
        \le 2b.
    \end{align*}
    Therefore we do not add any $e\in\cE(u_0)$ into $S$ in the first iteration. Suppose no $e\in\cE(u_0)$ has been added to $S$ during the first $i$ iterations, then before the $(i+1)^{th}$ iteration starts, $u_0$ has not been merged with any other vertex in the contracted graph $G/S$, so $\cE(u_0)$ in $G/S$ is the same as the original graph $G$. Thus $\sum_{e\in\cE(u_0)}x_e\le 2b$ still holds for $u_0$ in $G/S$, and by the same argument as the first iteration, no $e\in\mathcal E(u_0)$ will be added to $S$ in the $(i+1)^{th}$ iteration. 
\end{proof}

Since $\bx\in b\cP_M$, we have $\sum_{e\in E}x_e\le b(n-1)$, which implies $\sum_{u\in V}\sum_{e\in\cE(u)}x_e\le 2b(n-1)$. By averaging, there exists a vertex $u_0\in V$ such that $\sum_{e\in\cE(u_0)} x_e\le 2b (n-1)/n \le2b$, and by \Cref{clm:graphic-good-vertex}, $\cE(u_0)\cap E_1=\varnothing$. Assuming no isolated vertex in $V$, $E_1\subsetneq E_0$. Similarly, for any $i$, since $M|_{E_i}$ is
also a graphic matroid and $\bx|_{E_i}\in b\cP_{M|_{E_i}}$, the same argument holds for it. Therefore
$E_{i+1}\subsetneq E_i$ always holds, which finishes our proof of
termination for our construction. \end{proof}
We show limitations of our approach above in \cref{sec:limit-graphic}.

\section{Transversal Matroids}
\label{sec:transversal}

In this section we show a $(b,1-b)$-selectable PI-OCRS for transversal matroids. Recall that a transversal matroid is represented by a bipartite graph $G=(U,V,E)$, with vertex sets $U,V$ and edge set $E\sse U\times V$. Let $M=(U,\cI)$ denote the transversal matroid, and let $\cD$ be a distribution over $U$ that is PI-consistent with some $\bx\in b\cP_M$. Let $\mathcal N(i)$ denote the neighbour vertices of vertex $i$. The OCRS works as follows: first, for each edge $(i,j)\in E$, we assign some probability $y_{i,j}\in[0,1]$ to it, such that
\begin{align}
    \sum_{j\in\mathcal N(i)}y_{i,j}&=1 \qquad \forall i\in U, \label{eqmatching1}\\
    \sum_{i\in\mathcal N(j)}x_i \: y_{i,j}&\le b 
 \qquad \forall j\in V. \label{eqmatching2}
\end{align}
These $y_{i,j}$ exist (and can be efficiently computed) by the following lemma:
\begin{lemma}\label{lemmamatch}
    Let $G=(U,V,E)$ denote a bipartite graph, and $M$ the transversal matroid defined on $U$. Then for any $x \in b\cP_M$, there exist corresponding $y_{i,j} \in [0,1]$ satisfying \eqref{eqmatching1} and \eqref{eqmatching2}, and such $y_{i,j}$ can be computed efficiently.
\end{lemma}

For each left vertex $i\in U$, let $Y_i\in \mathcal N(i)$ be a random variable such that $\Pr[Y_i=j]=y_{i,j}$. Here $\{Y_i\}_{i\in U}$ are mutually independent, and are independent from the events $\{i\in R\}_{i\in U}$. Then the feasible sets $\cF$ are defined as
$$
    \cF=\{S\sse U\mid Y_i\ne Y_j\quad\forall i,j\in S,i\ne j\}
$$
It's not hard to see that $\cF\sse\cI$, because for any $S\in\cF$, every $i\in S$ is matched to $V$ by $Y_i$. It remains to bound the selectability of $\cF$. For any $i\in U$, we have
\begin{align*}
    &\Pr[I\cup\{i\}\in\cF\quad\forall I\in\cF,I\sse R\mid i\in R]\\
    &=\sum_{j\in\mathcal N(i)}\Pr[Y_i=j]\Pr[I\cup\{i\}\in\cF\quad\forall I\in\cF,I\sse R\mid i\in R,Y_i=j]\\
    &=\sum_{j\in\mathcal N(i)}y_{i,j}\Pr[Y_{i'}\ne j\quad\forall i'\in R\setminus\{i\}\mid i\in R]\\
    &\ge\sum_{j\in\mathcal N(i)}y_{i,j}\bigg(1-\sum_{i'\in\mathcal N(j)\setminus\{i\}}\Pr[Y_{i'}=j\land i'\in R\mid i\in R]\bigg)\tag{Union bound}\\
    &=\sum_{j\in\mathcal N(i)}y_{i,j}\bigg(1-\sum_{i'\in\mathcal N(j)\setminus\{i\}}x_{i'}y_{i',j}\bigg)\tag{Pairwise independence of $i\in R$}\\
    &\ge\sum_{j\in\mathcal N(i)}y_{i,j}(1-b)\tag{Inequality \eqref{eqmatching2}}\\
    &=1-b.\tag{Equation \eqref{eqmatching1}}
\end{align*}
Therefore we conclude this algorithm is indeed a $(b,1-b)$-selectable PI-OCRS, and by \Cref{clm:scaleOCRS}, setting $b=\nf12$ gives a $\nf14$-selectable PI-OCRS.

\medskip
\begin{proof}[Proof of \Cref{lemmamatch}]
    We give an algorithmic proof. First, construct a network flow model $G'$ based on $G$ by adding a source vertex $s$ and a sink vertex $t$. For every $i\in U$ add an arc from $s$ to $i$ with capacity $x_i$, and for every $j\in V$ add an arc from $j$ to $t$ with capacity $1$. Then for each edge $(i,j)\in E$, add the arc $(i,j)$ to $G'$ with infinite capacity. Then compute a maximum $s\rightarrow t$ flow.

    We claim that the maximum flow is equal to $\sum_{i\in U}x_i$. To see this, we argue that $\{s\}\times U$ is a minimum cut of the flow graph. Otherwise, let $C$ be a minimum cut where $C\ne \{s\}\times U$. Let $C_U=C\cap (\{s\}\times U)$ and $C_V=C\cap (V\times \{t\})$, and let $\overline C_U=(\{s\}\times U)\setminus C_U$. Here $\{j\in N(i)\mid (s,i)\in\overline C_U\}\times \{t\}\sse C_V$ must hold, since otherwise $C$ must contain infinite-capacity arcs. Then by the matroid polytope constraint, $\sum_{(s,i)\in \overline C_U}x_i\le \rank(\overline C_U)\le |C_V|$, meaning that the capacity of $\{s\}\times U=(C\setminus C_V)\cup\overline C_U$ is no larger than $C$. Therefore $\{s\}\times U$ is a minimum cut. And by the max-flow min-cut theorem, the maximum flow is equal to $\sum_{i\in U}x_i$.

    Then finally set $y_{i,j}=f_{i,j}/x_i$, where $f_{i,j}$ is the flow on edge $(i,j)$ in the maximum flow. The fact that $f$ satisfies the capacity constraints on $s \times U$ and $V \times t$ implies \eqref{eqmatching1} and \eqref{eqmatching2}, respectively.
\end{proof}

\section{Cographic Matroids}
\label{sec:cographic}

In this section we give a PI-OCRS for cographic matroids.

\begin{restatable}[Cographic Matroids PI-OCRS]{theorem}{cographicCRS}
  \label{thm:cographic}
  There is a $\nf{1}{12}$-selectable PI-OCRS for cographic matroids.
\end{restatable}

The idea is similar to \cite{soto2013matroid}. We first transform a given cographic matroid into a ``low density'' matroid by removing the parallel items, then combine the OCRS for low density matroids with a single-item OCRS to get an OCRS for cographic matroids. We first define the density of a matroid.

\begin{definition}
    Let $M=(E,\mathcal I)$ be a loopless matroid (i.e. all circuits are of size $\ge2$). Then the \emph{density} of $M$ is defined as
    $$
        \gamma(M):=\max_{S\sse E,S\ne \varnothing}\frac{|S|}{\rank(S)}.
    $$
\end{definition}

Just like the matroid secretary problem \cite{soto2013matroid}, we have a good PI-OCRS for low-density matroids.

\begin{lemma}\label{lem:lowdenOCRS}
    There is a $\nicefrac{1}{\gamma(M)}$-selectable PI-OCRS for matroids with density $\gamma(M)$.
\end{lemma}

\begin{proof}
    Let $\mathbf y\in\mathbb R^E$ be a vector with all its coordinates equal to $\frac1{\gamma(M)}$. Then for any $S\subseteq E$, by the definition of density, $\mathbf y(S)\le\rank(S)$ holds. Therefore $\mathbf y\in \cP_{M}$, and hence we can write $\mathbf y=\sum_{I\in\mathcal I} p_I\chi_I$. The PI-OCRS $\pi$ works as follows. Before items start arriving, $\pi$ samples an independent set $I_0$ according to distribution $\{p_I\}_I$, and set $\cF=2^{I_0}$ for the greedy PI-OCRS. By definition $\cF\sse\cI$. For any item $i\in E$, we have
    \begin{align*}
        \Pr[I\cup\{i\}\in\cF\quad\forall I\in\cF,I\sse R\mid i\in R] &=\Pr[I\cup\{i\}\sse I_0\quad\forall I\in R\cap I_0\mid i\in R]\\
        &=\Pr[i\in I_0] =\sum_{I\ni i}p_I =\frac1{\gamma(M)}.
    \end{align*}
    Therefore this PI-OCRS is $1/\gamma(M)$-selectable.
\end{proof}

Let $M$ be a loopless cographic matroid associated with graph $G$. $M$ is not necessarily low-density: for example, if $G$ is a cycle of length $n$, the density of $M$ is $n$. Therefore to give a PI-OCRS for cographic matroids, first we show that if each vertex of $G$ has degree at least $3$, the density $M$ is at most $3$; Then we show how to reduce arbitrary graphs to a graph where degree of each vertex is at least $3$.

\begin{lemma}\label{lemmadensity}
    Let $G=(V,E)$ be a graph such that $\text{deg}(v)\ge 3$ for all $v\in V$. Then the cographic matroid $M$ associated with $G$ has density $\gamma(M)\le 3$.
\end{lemma}
\begin{proof}
    Since each vertex has degree at least $3$, we have $|E|\ge \frac32|V|$. Also since $M^*$ (the dual matroid of $M$) is a graphic matroid on $G$, we have $\rank(M^*)\le|V|-1$, and thus $\rank(M)=|E|-\rank(M^*)\ge|E|-|V|+1$. Then we have
    $$
        \frac{|E|}{\rank(M)}\le\frac{|E|}{|E|-|V|+1}\le\frac{|E|}{\frac13|E|}=3
    $$
    Now for any $E'\subset E$, we want to give an upper bound for $\frac{|E'|}{\rank(E')}$. Let $M':=M|_{E'}$, then the density of $E'$ is equal to $\frac{|E'|}{\rank(M')}$. Since deletion in cographic matroids is equivalent to contraction in graphic matroid, $M'$ is associated with the graph $G'$ obtained by contracting $E\setminus E'$ in $G$. It's not hard to see that each vertex in $G'$ also has degree at least $3$. Therefore the above argument about density also holds for $M'$, i.e. $\frac{|E'|}{\rank(M')}\le 3$. Therefore $\gamma(M)\le 3$.
\end{proof}

\begin{proof}[Proof of \Cref{thm:cographic}]
    By \Cref{lem:lowdenOCRS} and \Cref{lemmadensity}, we already have a $\nf13$-selectable PI-OCRS for cographic matroids, where each vertex of the underlying graph has degree at least $3$. We now consider the case of general graphs. Given a cographic matroid $M=(E,\cI)$, we say two elements $a,b\in E$ are parallel if $a=b$ or $\{a,b\}$ is a circuit of size $2$. We say two sets $A,B\sse E$ are parallel if there exists a bijection $\phi:A\to B$ such that for any $a\in A$, $a$ is parallel to $\phi(a)$. Then one can verify that parallelism defines an equivalence relation over $E$, and thus $E$ can be partitioned into a disjoint union of equivalence classes: $E=\bigsqcup_{j=1}^m E_j$. For each equivalence class $E_j$, we select a representative $e_j$, and define $E':=\{e_j:j\in[m]\}$. We then define $M'=M|_{E'}$, then $M'$ is also a cographic matroid, let $G'=(V',E')$ denote the underlying graph for $M'$. Assuming no isolated vertex in $G'$, since $M'$ contains no parallel elements, the degree of each vertex in $G'$ is at least $3$. Therefore by \Cref{lem:lowdenOCRS} and \Cref{lemmadensity}, we have a $1/3$-selectable PI-OCRS for $M'$. Let $I_0\sse E'$ denote the independent set selected in \Cref{lem:lowdenOCRS}. Now for any $\bx\in b\cP_M$ and distribution $\cD$ PI-consistent with $\bx$, we set $\cF=\{I\sse E\mid I\text{ is parallel to some } I'\sse I_0\}$. Then for any $e\in E_j$, we have
    \begin{align*}
        \Pr[I\cup \{e\}\in\cF\quad\forall I\in\cF,I\sse R\mid e\in R] &=\Pr[E_j\cap R=\{e\}\mid e\in R]\cdot \Pr[e_j\in I_0]\\
        &\ge\frac13\Pr\bigg[\bigwedge_{e'\in E_j\setminus\{e\}}e'\notin R\mid e\in R\bigg]\tag{\Cref{lem:lowdenOCRS}}\\
        &\ge\frac13\bigg(1-\sum_{e'\in E_j\setminus\{e\}}\Pr[e'\in R]\bigg)\tag{union bound and pairwise independence}\\
        &\ge\frac13(1-b)\tag{$\bx(E_j)\le b\rank(E_j)=b$}
    \end{align*}
    Therefore we have a $(b,\frac13(1-b))$-selectable PI-OCRS. By \Cref{clm:scaleOCRS}, setting $b=\nf{1}{2}$ gives a $\nf{1}{12}$-selectable PI-OCRS for cographic matroids.
\end{proof}

\section{Regular Matroids}
\label{sec:regular}

We now give a $\Omega(1)$-competitive PI-OCRS for regular matroids. We
use the regular matroid decomposition theorem of
\textcite{seymour1986decomposition} and its modification by
\textcite{dinitz2014matroid}, which decomposes any regular matroid into
1-sums, 2-sums, and 3-sums of graphic matroids, cographic matroids,
and a specific 10-element matroid $R_{10}$. (These matroids are called
the \emph{basic} matroids of the decomposition). We now define
$1$,$2$,$3$-sums, and argue that it suffices to run a PI-OCRS for each of the basic matroids and to output the union of their outputs.  

\begin{definition}[Binary Matroid Sums \cite{dinitz2014matroid,seymour1986decomposition}]
  Given two binary matroids $M_1=(E_1,\mathcal I_1)$ and
  $M_2 = (E_2,\mathcal I_2)$, the \emph{matroid sum} $M$ defined on the ground set
  $E(M_1) \Delta E(M_2)$ is as follows. The set $C$ is a cycle
  in $M$ iff it can be written as $C_1\Delta C_2$, where $C_1$ and
  $C_2$ are cycles of $M_1$ and $M_2$, respectively. Furthermore,
\begin{enumerate}
	\item If $E_1\cap E_2=\varnothing$, then $M$ is called 1-sum
          of $M_1$ and $M_2$. 
	\item If $|E_1\cap E_2|=1$, then we call $M$ the 2-sum of
          $M_1$ and $M_2$. 
	\item If $|E_1\cap E_2|=3$, let $Z=E_1\cap E_2$. If $Z$ is a
          circuit of both $M_1$ and $M_2$, then we call $M$ the 3-sum
          of $M_1$ and $M_2$. 
\end{enumerate}
\end{definition}

(The $i$-sum is denoted $M_1 \oplus_i M_2$.) Our definition differs from \cite{seymour1986decomposition,dinitz2014matroid} as we have dropped some conditions on the sizes of $M_1$ and $M_2$ that we do not need.
A $\{1,2,3\}$-\emph{decomposition} of a matroid $\widetilde M$ is a set of matroids $\cM$ called the \emph{basic
matroids}, together with a rooted binary tree $T$ in which $\widetilde M$ is the root and the leaves are the elements of $\cM$. Every internal vertex in the tree is either the $1$-, $2$-, or $3$-sum of its children.
Seymour's decomposition theorem for regular matroids
\cite{seymour1986decomposition} says that 
  every regular matroid $\widetilde M$ has a (poly-time computable) $\{1,2,3\}$-decomposition
  with all basic matroids being graphic, cographic or
  $R_{10}$. 

\paragraph{The Dinitz-Kortsarz modification.} \textcite{dinitz2014matroid} modified Seymour's
decomposition to give an $\Omega(1)$-competitive algorithm for the
regular-matroid secretary problem, as follows. Given a $\{1, 2, 3\}$-decomposition $T$ for binary matroid $\widetilde M$ with basic matroids $\cM$, we define $Z_M$, the \emph{sum-set} of a non-leaf vertex $M$ in $T$, to be the intersection of the ground sets of its children (the sum-set is thus not in
the ground set of $M$). A sum-set $Z_M$ for internal vertex $M$ is either the empty set (if $M$ is the $1$-sum of its children), a single element (for $2$-sums), or three elements that form a circuit in its children (for $3$-sums). A $\{1, 2, 3\}$-decomposition is \emph{good} if for every sum-set $Z_M$ of size $3$ associated with internal vertex $M = M_1 \oplus_3 M_2$, the set $Z_M$ is contained in the ground set of a single basic matroid below $M_1$, and in the ground set of a single basic matroid below $M_2$.
For a given $\{1,2,3\}$-decomposition of a matroid $\widetilde M$ with
basic matroids $\mathcal{M}$, define the \emph{conflict graph} $G_T$
to be the graph on $\cM$ where basic matroids $M_1$ and $M_2$ share an edge if their ground sets intersect.
\cite{dinitz2014matroid} show that 
    if $T$ is a good $\{1, 2, 3\}$-decomposition of $\widetilde M$, then $G_T$ is a forest.
We can root each tree in such a forest arbitrarily, and define the
parent $p(M)$ of each non-root matroid $M \in \cM$. Let $A_M$ be the sum-set
for the edge between matroid $M$ and its parent, i.e.,
$A_M = E(M) \cap E(p(M))$.
\begin{theorem}[Theorem 3.8 of \cite{dinitz2014matroid}]
  \label{thm:good_123decomp}
  There is a good $\{1, 2, 3\}$-decomposition $T$ for any binary matroid $\widetilde M$ with basic matroids $\cM$ such that (a)~each basic matroid $M \in \cM$ has no circuits of size $2$ consisting of an element of $A_{M}$ and an element of $E(\widetilde
  M)$, and (b)~every basic matroid $M \in \cM$ can be obtained from some $M' \in \widetilde \cM$ by deleting elements and adding parallel elements.
\end{theorem}
Dinitz and Kortsarz showed that a good $\{1,2,3\}$-decomposition for a matroid $\widetilde M$ can be used to construct independent sets for $\widetilde M$ as follows. Below, $\cdot \vert_S$ denotes restriction to the set $S$.
\begin{lemma}[Lemma 4.4 of \cite{dinitz2014matroid}]
\label{lem:gluing_the_pieces}
    Let $T$ be a good $\{1,2,3\}$-decomposition for $\widetilde M$ with basic matroids $\cM$. For each each $M \in \cM$, let $I_M$ be an independent set of $(M/A_M )\vert_{(E(M) \cap
E(\widetilde M))}$. 
Then $I = \bigcup_{M \in \cM} I_M$ is independent
in $\widetilde M$. 
\end{lemma}
\paragraph{Our Algorithm.} Given the input matroid $\widetilde M$, our
idea is to take a good decomposition $T$ and run a PI-OCRS for
$(M/A_M )\vert_{(E(M) \cap E(\widetilde M))}$ for each vertex $M$ in the
conflict graph $G_T$. Then we need to glue the pieces together using
\cref{lem:gluing_the_pieces}. 
  One technical obstacle is  that the input to an OCRS is a feasible point in the matroid polytope, so to use the framework of \cite{dinitz2014matroid}  we need to convert it into a feasible solution to the polytopes of the (modified) basic matroids. Our insight is captured by the following lemma.
\begin{lemma}
    \label{lem:projecting_x}
    Let $T$ be a good $\{1,2,3\}$-decomposition of regular matroid $\widetilde M$ with basic matroids $\cM$, and let vector $\mathbf{x} \in \frac{1}{3}\cP_{\widetilde M}$. Then for every basic matroid $M \in \cM$, if $\widehat M := (M/A_M )\vert_{(E(M) \cap
E(\widetilde M))}$, then $\mathbf{x} \vert_{\widehat M} \in \cP_{\widehat M}$.
\end{lemma}
\begin{proof}
    Fix a set $S \subseteq E(M) \cap
E(\widetilde M)$. We will show that $\rank_{\widehat M}(S) \geq \frac{1}{3} \rank_{\widetilde M}(S)$, from which the claim follows. 

\textbf{Case 1}: $A_M = \{z\}$. For any maximal independent set $I\subset S$ according to $M$, there always exists $a \in I$ such that $(I\cup \{z\})\setminus \{a\}$ is independent in $M$, therefore $\rank_{\widehat M}(S)\geq\rank_{\widetilde M}(S)-1$. Also since no element in $S$ is parallel to $z$, for any non-empty $S$  we have $\rank_{M/A_M}(S)\geq 1$, and we can conclude that $\rank_{M/A_M}(S)\ge\frac{1}{3}\rank_{\widetilde M}(S)$.

\textbf{Case 2}: $A_M$ is some 3-cycle $\{z_1, z_2, z_3\}$. For any maximal independent set $I\subset S$ according to $M$ where $|I| \geq 3$, there always exists $a,b\in I$ such that $(I\cup \{z_1,z_2\})\setminus\{a,b\}$ is independent in $M$. Therefore $\rank_{M_2/Z}(S)\ge\rank_M(S)-2$. We claim that there does not exist $e$ in $E(M)\setminus A_M$ such that $e \in \spn(A_M)$.

Suppose for contradiction such an $e$ exists. Then there is some circuit in $A_M\cup e$ containing $e$. Since there are no parallel elements, this circuit have size $3$. Without loss of generality, assume this circuit is $C=\{z_1,z_2,e\}$. Since $A_M$ is a circuit, by definition of binary matroids, the set $C\Delta A_M=\{z_3,e\}$ is a cycle, and thus $e$ is parallel to $z_3$, a contradiction. Therefore for any non-empty $S$, we have that $\rank_{M/A_M}(S)\geq 1$, and we conclude that $\rank_{M/A_M}(S)\ge\frac{1}{3}\rank_{\widetilde M}(S)$.
\end{proof}

Next, we argue that for every $M$ in the conflict graph, there exists a suitable PI-OCRS for $\widehat M := (M/A_M)\vert_{(E(M) \cap
 E(\widetilde M))}$. For a matroid $M$, let $S(M)$ denote the set of matroids that can be obtained from $M$ by iteratively removing elements, adding parallel elements, and contracting either nothing, a single element, or a circuit of size $3$.

\begin{lemma}
    \label{cor:ocrs_for_sm}
    If matroid $M$ is either graphic, cographic, or $R_{10}$, then every matroid $M' \in S(M)$ admits a $\nf{1}{12}$-selectable PI-OCRS.
\end{lemma}

\begin{proof}
    The class of graphic (resp. cographic) matroids are closed under the operations of deletion, contraction, and addition of parallel elements. Hence if $M$ is graphic or cographic, so are all matroids in $S(M)$.

    Otherwise, $M = R_{10}$. It is known (see the proof of \cite[Corollary 4.7]{dinitz2014matroid}) that for any matroid $M' \in S(R_{10})$, either $M'$ is graphic, cographic, or has a ground set can be partitioned into $10$ parts such that deleting any one of the $10$ parts makes $M$ graphic.

    If $M' \in S(M)$ is graphic or cographic, \cref{thm:graphic,thm:cographic} imply a $\nicefrac{1}{12}$-selectable OCRS. Else, deleting one of the 10 parts of $M'$ uniformly at random and then running a graphic matroid OCRS yields a $(b, \frac{9}{10}\cdot (1-2b))$-selectable PI-OCRS, which can be tuned to give a $\nf{9}{80}$-selectable PI-OCRS. In conclusion, each matroid $M'\in S(M)$ admits a $\min(\nf1{12},\nf9{80})=\nf1{12}$-selectable PI-OCRS.
\end{proof}

Finally, we have all we need to prove the main theorem of this section.

\begin{restatable}[Regular Matroids]{theorem}{regmat}
    \label{thm:regmat}
      There is a $(\nf13,\nf1{12})$-selectable PI-OCRS for regular matroids.
\end{restatable}

\begin{proof}
    
The input is a regular matroid $\widetilde M$ and any distribution $\cD$ over $E(M)$ PI-consistent with some vector $\mathbf{x} \in \frac{1}{3}\cP_{\widetilde M}$.

Our greedy PI-OCRS for regular matroid $\widetilde M$ proceeds as follows. First, compute a good $\{1,2,3\}$-decomposition $T$ of $\widetilde M$ with basic matroid set $\cM$, using \cref{thm:good_123decomp}. For each of the basic matroids $M \in \cM$, recall that $\widehat M := (M/A_M)\vert_{(E(M) \cap
E(\widetilde M))}$. By \cref{lem:projecting_x}, $\mathbf{x} \vert_{\widehat M} \in \cP_{\widehat M}$. Since $M$ is graphic/cographic/$R_{10}$ and $\widehat M \in S(M)$, by \cref{cor:ocrs_for_sm} there is a $\nf{1}{12}$-selectable PI-OCRS for each $\widehat M$, and we denote the feasible sets of such greedy PI-OCRS as $\cF_{\widehat M}$. Then we simply set $\cF_{\widetilde M}=\{I\sse \widetilde M\mid I\cap E(\widehat M)\in\cF_{\widehat M},\forall M\in\cM\}$. By \Cref{lem:gluing_the_pieces}, $\cF_{\widetilde M}\sse\cI(\widetilde M)$, i.e. $\cF_{\widetilde M}$ is indeed feasible. Then we bound the selectability of the greedy PI-OCRS defined by $\cF_{\widetilde M}$: for any item $i$ in $\widehat M$, we have
\begin{align*}
    &\Pr[I\cup\{i\}\in\cF_{\widetilde M}\quad\forall I\in\cF_{\widetilde M},I\sse R\mid i\in R]\\
    &=\Pr[I\cup\{i\}\in\cF_{\widehat M}\quad\forall I\in\cF_{\widehat M},I\sse R\cap\widehat M\mid i\in R\cap\widehat M] \ge\frac1{12}.\tag{\Cref{lem:projecting_x}}
\end{align*}
We conclude that $\cF_{\widetilde M}$ does indeed give a $(\nf{1}{3},\nf{1}{12})$-selectable PI-OCRS for regular matroids. \qedhere
\end{proof}

\printbibliography

\appendix

\section{The Single-Item Setting}
\label{sec:single-item}

\subsection{A Single-Sample Prophet Algorithm} 
\label{sec:single-sample}

In this section we give an $\Omega(1)$ competitive algorithm for the prophet problem under a pairwise independent distribution, in the restricted setting where the algorithm has only limited prior knowledge of the pairwise independent distribution $\cD$ in the form of a single sample from $\cD$.

Formally, the algorithm is given as input $\langle \widetilde X_1, \ldots, \widetilde X_n \rangle \sim \cD$. Subsequently, the values of a second draw $\langle X_1, \ldots, X_n \rangle \sim \cD$ are revealed one by one in order. The algorithm's goal is to stop at the largest value. We show the following theorem.

\singlesample*

\begin{proof}
    Our algorithm sets a threshold $v^*= \max_i \widetilde X_i$ (which is a random variable), and then stops at the first $X_i \geq v^*$ if such an $X_i$ exists.

    To understand the performance of this algorithm, define the quantity $x(t) := \sum_i \Pr[\widehat v_i \geq t]$. The proof of \cite[Theorem 2]{DBLP:conf/wine/CaragiannisGLW21} shows that an algorithm using threshold $t$ achieves competitive ratio $\min(1-t,\: t/(t+1))$. Letting $f$ be the probability density function of $x(v^*)$, and defining
    \[F(t):=\int_0^tf(r)\text dr = \Pr[x(v^*)\le t] = \Pr[v^*\ge x^{-1}(t)],\] 
    we can lower bound our competitive ratio as
\begin{align*}
\text{APX}&\ge\int_0^1f(t)\min\left(1-t,\frac t{1+t}\right)\text dt =\int_0^{\frac{\sqrt5-1}2}f(t)\frac t{1+t}\text dt+\int_{\frac{\sqrt5-1}2}^1f(t)(1-t)\text dt.
\end{align*}
We apply integration by parts to the first term to get
\begin{align*}
    \int_0^{\frac{\sqrt5-1}2}f(t)\frac t{1+t}\text dt &=F(t)\frac t{1+t}\big|_0^{\frac{\sqrt5-1}2}-\int_0^{\frac{\sqrt5-1}2}f(t)\frac1{(1+t)^2}\text dt\\
    &=\frac{3-\sqrt5}2F\left(\frac{\sqrt5-1}2\right)-\int_0^{\frac{\sqrt5-1}2}F(t)\frac1{(1+t)^2}\text dt,
\end{align*}
and to the second term which gives
\begin{align*}
    &\int_{\frac{\sqrt5-1}2}^1f(t)(1-t)\text dt =F(t)(1-t)\big|_{\frac{\sqrt5-1}2}^1+\int_{\frac{\sqrt5-1}2}^1F(t)\text dt =-\frac{3-\sqrt5}2F\left(\frac{\sqrt5-1}2\right)+\int_{\frac{\sqrt5-1}2}^1F(t)\text dt.
\end{align*}
Since $F(t)=\Pr[v^*\ge x^{-1}(t)]=\Pr[\bigvee_{i=1}^nv_i\ge x^{-1}(t)]$, by \cite[Lemma 1]{DBLP:conf/wine/CaragiannisGLW21}, we have $\frac t{1+t}\le F(t)\le t$. Thus
\begin{align*}
    \text{APX}&\ge\int_{\frac{\sqrt5-1}2}^1F(t)\text dt-\int_0^{\frac{\sqrt5-1}2}F(t)\frac1{(1+t)^2}\text dt \ge\int_{\frac{\sqrt5-1}2}^1\frac t{1+t}\text dt-\int_0^{\frac{\sqrt5-1}2}\frac t{(1+t)^2}\text dt\\
    &=3-\sqrt5-\ln2.
\end{align*}
Therefore this algorithm gives a competitive ratio of $(3-\sqrt5-\ln2)$, or approximately $7\%$.
\end{proof}
\subsection{Bounds for $t$-Wise Independent CRSs}
\label{sec:twise}

We know that when $\cD$ is independent, $c$-balanced single-item CRSs exist for $c = 1-(1-1/n)^n \approx 1-1/e$, and that this is tight \cite{DBLP:conf/focs/FeigeV06,DBLP:journals/toc/FeigeV10}. Indeed \textcite{DBLP:journals/siamcomp/ChekuriVZ14} show that a $c$-balanced CRS for $c=1-(1-1/n)^n$ exists for any matroid.
In this section we ask: what is the best possible $c$ for merely $t$-wise independent distributions $\cD$? We will show:

\twise*

We begin with an upper bound.

\subsubsection{Upper Bounds for $t$-Wise Independent CRSs}

As a warm-up, consider the case of 2- and 3-wise independent $\cD$. 
\begin{theorem} \label{thm:23wise-CRS-upperbound}
    For $2$- or $3$-wise independent $\cD$, no $c$-balanced CRS exists for $c > \frac{1}{2} + O(1/n)$.
\end{theorem}

We demonstrate this and the subsequent claims by constructing, for each $n$, distributions for which no CRS can do better than the stated bound. We will consider distributions $\cD$ that are \emph{symmetric}, meaning that all equal-size subsets are equally likely; for all $S, S' \subseteq [n]$ if $|S|=|S'|$ then $\Pr[R=S] = \Pr[R=S']$.

For symmetric distribution $\cD$ and a given $n$, let $z_i$ for $i \in \{0, 1, \ldots, n\}$ be the probability that $\cD$ chooses exactly $i$ items from $[n]$ to be active. Then the map that sends vectors $z = (z_0, \ldots, z_n) \in \R_+^{n+1}$ for which $\|z\|_1 = 1$ to the symmetric distributions $\cD$ that chooses $S \subseteq [n]$ to be active with probability exactly $z_{|S|} / {[n] \choose |S|}$ is a bijection. Note that the marginal probabilities for $i\in R$ and symmetric distribution $\cD$ are given by $\Pr_{R \sim \cD}[i\in R] = \sum_j \nicefrac{j}{n} \cdot z_j$, and so in particular $\Pr_{R \sim \cD}[i\in R] = 1/n$ when $\sum_j j\cdot z_j = 1$. 

For such distributions $\cD$, we can provide an upper bound on the performance of any CRS in terms of these $z_i$.
\begin{claim} \label{lem:CRSfromZbound}
    Let $\cD$ be a symmetric distribution given by $z$, with marginals $x_i = 1/n$. Then no CRS is more than $(1-z_0)$-balanced on $\cD$.
\end{claim}
\begin{proof}
    Fix a CRS that picks the set $I \subseteq R$. The events $\{\{i \in I \}\}_i$ are disjoint, and only occur if at least one element is active, which happens with probability at most $1-z_0$. By averaging, for at least one $i \in [n]$ we have $\Pr_{R \sim \cD}[i \in I] \leq (1-z_0)/n$. Consequently, using $x_i = 1/n$ we have $\Pr_{R \sim \cD}[i \in I \mid i \in R] = \Pr_{R \sim \cD}[i \in I] / \Pr_{R \sim \cD}[i \in R] \leq 1 - z_0$. In other words the CRS is no more than $(1-z_0)$-balanced.
\end{proof}

Having established this, we are ready to demonstrate our $2$- and $3$-wise independent hard distributions.
\begin{proof}[Proof of \Cref{thm:23wise-CRS-upperbound}]
    We begin with the pairwise-independent symmetric $\cD$. It is given by 
    \[
        z_0 =\frac{n-1}{2n}, \qquad z_1 = \frac{1}{n}, \qquad z_2 = \frac{n-1}{2n}, \qquad z_j = 0 \quad \forall j \geq 3.
    \]
    One can verify that this $z$ satisfies $\|z\|_1 = 1$ (so it describes a distribution), and that it is $2$-wise independent. Furthermore $\sum_j j \cdot z_j = 1$, and so the marginals are $x_i = 1/n$. On the other hand, by \Cref{lem:CRSfromZbound} no CRS can do better than $1-z_0 = \frac{1}{2} + O(1/n)$ on this instance.

    \medskip
    We may slightly modify the above example in order to create a hard $3$-wise independent instance by choosing $R = [n]$ with some small probability, as follows:
    \begin{align*}
                &z_0 =\frac{n^3 - 2n^2 + 3n - 2}{2n^3}, \qquad z_1 = \frac{2n-2}{n^2}, \qquad z_2 = \frac{(n - 1)^2}{2n^2}, \qquad z_n = \frac{1}{n^3}, \\
                &z_j = 0 \quad \forall j \in \{3 ,\ldots, n-1\}.
    \end{align*}
    
Again this satisfies $\|z\|_1 = 1$ and $\sum_j j \cdot z_j = 1$, and so $x_i = 1/n$ for each $i \in [n]$. One can again verify that this is $3$-wise independent, and so by \Cref{lem:CRSfromZbound} this gives a CRS upper bound of $1-z_0 = \frac{1}{2} + O(1/n)$.
\end{proof}

What about $t$-wise independent hard instances for $t \geq 4$? More generally the performance ratio follows the sequence
\begin{equation}
    \phi(t) := 1 - \sum_{k=0}^t \frac{(-1)^k}{k!}, \label{eq:twiseapproxseries}
\end{equation}
where $\phi(2) = 1/2$ matches the guarantee of \Cref{thm:23wise-CRS-upperbound} for pairwise-independent distributions ($t=2$).
As expected $\lim_{t \rightarrow \infty} \phi(n) = 1-1/e$, and the $\phi(t)$ for even $t$ are underestimates of $1-1/e$, while odd $t$ give overestimates.

In general, for each $t \geq 1$ let $\cD_{t,n}$ be the unique symmetric $t$-wise independent distribution which chooses $R \sim [n]$ such that $|R| \leq t$ and for which the marginals are $x_i = 1/n$. Then we have the following claim.
\begin{lemma} \label{lem:Dtn-instance-exists}
    For each $t \geq 1$ and sufficiently large $n \geq t$, the distribution $\cD_{t,n}$ is well-defined, and no $c$-balanced CRS for $\cD_{t,n}$ exists for $c > \phi(t) + O(1/n)$.
\end{lemma}
\begin{proof}
    We begin by considering the constraints which characterize $\ell$-wise independence for $\ell \leq t$. 
    The probability that any given $\ell$ items $i_1, \ldots, i_\ell$ are in $R \sim \cD_{t,n}$ is given by
    \begin{align*}
        \Pr_{R \sim \cD}[i_1, \ldots, i_\ell \in R] &= \sum_{j = \ell}^t \Pr[i_1, \ldots, i_\ell \in R \wedge |R|=j] \\
        &= \sum_{j = \ell}^t \Pr[i_1, \ldots, i_\ell \in R \mid |R|=j] \cdot \Pr[|R|=j] \\
        &= \sum_{j = \ell}^t {n-\ell \choose j - \ell}{n \choose j}^{-1} \cdot z_j \\
        &= \sum_{j = \ell}^t \frac{(n-\ell)!}{n!} \frac{j!}{(j-\ell)!} \cdot z_j.
    \end{align*}
    Then $\cD_{t,n}$ is $t$-wise independent so long as these probabilities are $1/n^\ell$. Letting $c_{\ell, j} := j!/(j-\ell)!$ and $G(n,\ell) := n^\ell (n-\ell)!/n!$, we may then say that a symmetric $\cD_{t,n}$ is $t$-wise independent so long as the $z_j$ satisfy
    \begin{equation}\label{eq:tIconstraints}
        \sum_{j = \ell}^t G(n, \ell) \cdot c_{\ell,j} \cdot z_j = 1, \qquad \ell \in \{0, \ldots, t\}.
    \end{equation}
    Note that the constraint for $\ell = 0$ specifies that $\cD_{t,n}$ is a distribution, and the constraint for $\ell = 1$ specifies that the marginals $x_i = 1/n$ uniformly.

    The next observation is that for fixed $\ell$, $\lim_{n \rightarrow \infty} G(n, \ell) = 1$. So for a given $t$ we may take the limit as $n$ becomes large and consider the following set of simplified independence constraints
    \begin{equation}\label{eq:tIconstraintslimit}
        \sum_{j = \ell}^t c_{\ell,j} \cdot z_j = 1, \qquad \ell \in \{0, \ldots, t\}.
    \end{equation}
    We will first demonstrate the solution to this simplified set of constraints, then recover guarantees for finite $n$ from it.

    We claim that choosing $z_j^* = \frac{1}{j!} \sum_{b=0}^{t-j}\frac{(-1)^b}{b!}$ satisfies \eqref{eq:tIconstraintslimit}. To see this we evaluate the right-hand side of \eqref{eq:tIconstraintslimit} for a given $\ell$ as 
    \begin{align}
        \sum_{j = \ell}^t c_{\ell,j} \cdot z_j^* &= \sum_{j = \ell}^t \frac{j!}{(j-\ell)!} \cdot \frac{1}{j!} \sum_{b=0}^{t-j}\frac{(-1)^b}{b!} 
        =\sum_{j = \ell}^t \sum_{b=0}^{t-j}\frac{(-1)^b}{b!(j-\ell)!} = \sum_{a=0}^s \sum_{b=0}^{s-a} \frac{(-1)^b}{a! b!} =:F(s), \label{eq:Lwiseindepsimplified}
    \end{align}
    where \eqref{eq:Lwiseindepsimplified} is derived by substituting $s = t - \ell$ and $a = j - \ell$. We show that $F(s)$ is always equal to $1$ by induction on $s$. The base case of $s=0$ is clear. To induct, consider the difference
    \begin{align}
        F(s) - F(s-1) &= \sum_{a=0}^{s} \sum_{b=0}^{s-a} \frac{(-1)^b}{a! b!} - \sum_{a=0}^{s-1} \sum_{b=0}^{s-a-1} \frac{(-1)^b}{a! b!} \notag \\
        &= \sum_{b = 0}^{s} \frac{(-1)^b}{b! (s-b)!} \notag \\
        &= \frac{1}{s!} \sum_{b=0}^s (-1)^{b} {s \choose b} \notag \\
        &= \frac{1}{s!} \left(1 + (-1)^s + \sum_{b=1}^{s-1} (-1)^{b} \cdot \left({s-1 \choose b} + {s-1 \choose b-1} \right) \right) \notag \\
        &= \frac{1}{s!} \left(1 + (-1)^s + \sum_{b=1}^{s-1} (-1)^{b} {s-1 \choose b} + \sum_{b=1}^{s-1} (-1)^{b} {s-1 \choose b-1} \right) \notag \\
        &= \frac{1}{s!} \left(1 + (-1)^s -1 + (-1)^{s-1} \right) = 0, \notag 
    \end{align}
    where above we used that ${n \choose k} = {n-1 \choose k} + {n-1 \choose k-1}$ for $1 \leq k \leq n-1$ and cancelled all but the first and last terms of the two resulting sums. Since this difference is uniformly zero, $F(s) = 1$ for all $s \geq 0$, and so $z^*$ satisfies \eqref{eq:tIconstraintslimit} for all $t$. Additionally, since the system of equations \eqref{eq:tIconstraintslimit} is upper-triangular and none of the coefficients $c_{\ell,j}$ are zero, this solution $z^*$ is unique.

    We now return to \eqref{eq:tIconstraints} and the case of finite $n$. By Cramer's rule, the solution $z^*$ may be expressed as 
    \begin{equation}
        z_j^* = \frac{\det(C\vert_j)}{\det(C)}, \notag
    \end{equation}
    where $C$ represents the constraint matrix for the system of equations \eqref{eq:tIconstraintslimit} expressed as $Cz = 1$, and $C\vert_j$ denotes the matrix $C$ with the $j^{th}$ column replaced by $1$s.

    For the constraint matrix $C_n$ corresponding to the finite-$n$ constraints \eqref{eq:tIconstraints} the same is true, where 
    \begin{equation}
        z_j^n = \frac{\det(C_n\vert_j)}{\det(C_n)} \notag
    \end{equation}
    is the unique solution to \eqref{eq:tIconstraints} so long as $\det(C_n)$ is nonzero.
    
    Consider the coefficients $c_{\ell, j}^n = G(n, \ell) \cdot c_{\ell, j}$ of \eqref{eq:tIconstraints}, and note that for fixed $t$, 
    \[
        1 \leq G(n, \ell) \leq \left(\frac{n}{n-\ell}\right)^\ell \leq \left(\frac{n}{n-t}\right)^t = 1 + O(1/n).
    \]
    The Leibniz formula for the determinant expresses $\det(C^n\vert_j)$ and $\det(C^n)$ as polynomials in $c^n_{\ell,j}$ of degree $t$, so for fixed $t$ we have that $\det(C^n\vert_j) = \det(C\vert_j) \pm O(1/n)$ and $\det(C^n) = \det(C) \pm O(1/n)$. Since the $z_j^*$ are well-defined $\det(C) \neq 0$, so for sufficiently large $n$ 
    \begin{equation}
        z_j^n = \frac{\det(C_n\vert_j)}{\det(C_n)} = \frac{\det(C\vert_j)}{\det(C)} \pm O(1/n) = z_j^* \pm O(1/n) \notag
    \end{equation}
    is well-defined also. 
    
    To conclude we consider $z_0^n = z_0^* \pm O(1/n) = \sum_{b=0}^t \frac{(-1)^b}{b!} \pm O(1/n)$. The distribution $\cD_{t,n}$ satisfies the conditions of \Cref{lem:CRSfromZbound} by construction. Therefore by \Cref{lem:CRSfromZbound} no $c$-balanced CRS for $\cD_{t,n}$ exists for $c > 1 - z_0^n = \phi(t) + O(1/n)$, as claimed.
\end{proof}

Observe from \eqref{eq:twiseapproxseries} that $\phi(t) > 1-1/e$ for odd $t$, while the $\phi(t)$ for even $t$ approach $1-1/e$ from below. For odd $t = 2k + 1$ we may derive a $t$-wise independent distribution from $\cD_{2k,n}$ without changing the CRS upper bound, giving our general upper bound for all $t$.
\begin{theorem}[$t$-wise Independent Prophet Inequality Upper Bound]
    \label{thm:twise-CRS-upperbound}
    For every integer $t \geq 2$, there are $t$-wise independent distributions for which no single item $c$-balanced CRS exist when $c > \phi(t) + O(1/n)$.
\end{theorem}
\begin{proof}
    The case of even $t = 2k$ is established by \Cref{lem:Dtn-instance-exists}. It remains to show that the odd $t=2k + 1$ cases may be reduced to the even $t=2k$ cases. This reduction will match the relationship between the $2$- and $3$-wise independent instances that prove \Cref{thm:23wise-CRS-upperbound}.

    To construct tighter symmetric distributions $\cD_{t,n}'$ for odd $t=2k+1$ we will choose the unique symmetric distribution with marginals $x_i = 1/n$ supported on $z_0, \ldots, z_{t-1}, z_n$ (as opposed to support on $z_0, \ldots, z_t$ as in \Cref{lem:Dtn-instance-exists}). For this choice of nonzero $z_j$ the constraints for $t$-wise independence (analogous to \eqref{eq:tIconstraints}) will be 
    \begin{align}
        \sum_{j = \ell}^{t-1} G(n, \ell) \cdot c_{\ell,j} \cdot z_j + G'(n, \ell) \cdot z_n &= 1, \qquad \ell \in \{0, \ldots, t-1\} \label{eq:tIconstraintsodd} \\
        G'(n,\ell) \cdot z_n &= 1, \qquad \ell = t \notag
    \end{align}
    where $G'(n, \ell) = n^\ell$ and the coefficient $c_{\ell,n} = 1$ implicitly. Since $\lim_{n \rightarrow \infty} G'(n, \ell) = \infty$, the simplified constraints will be identical to \eqref{eq:tIconstraintslimit} for $(t-1)$-wise independence, together with the constraint $z_n^* = 0$, and so the limiting solutions will be $z_j^* = \frac{1}{j!} \sum_{b=0}^{t-j-1}\frac{(-1)^b}{b!}$ for $0 \leq j \leq t-1$, together with $z_n^* = 0$.

    The argument that values $z_j^n$ satisfying \eqref{eq:tIconstraintsodd} exist, are well-defined, and asymptotically equal $z_j^* \pm O(1/n)$ proceeds along the same lines as the proof of \Cref{lem:Dtn-instance-exists}. Setting $z_n^n = 1/n^t$ introduces only a lower-order perturbation to the system of equations, since for all $\ell \leq t-1$ the additional term is only $G'(n, \ell) z_n^n = n^{t-\ell} \leq n^{-1}$ and therefore of lower order.

    Finally \Cref{lem:CRSfromZbound} again applies, and the value of $1 - z_0^n = 1 - \sum_{b=0}^{t-1}\frac{(-1)^b}{b!} \pm O(1/n) = \phi(t-1) + O(1/n)$ provides the stated guarantee for odd $t$.
\end{proof}

\subsubsection{Lower Bounds for $t$-wise Independent CRSs}
\label{sec:single-item-CRS}

We turn to proving a tight lower bound by giving a single-item CRS algorithm for $t$-wise independent distributions whose approximation ratio matches the upper bounds given in the previous section. Our first goal is to show the following lower bound for the probability of union of $t$-wise independent events, which we will soon use to estimate the approximation ratio of the best CRS.

\begin{lemma}\label{cor:lb-union}
    Let $A_1,\ldots,A_n$ be $t$-wise independent random events for some even integer $t$. Define $x_i:=\Pr[A_i]$, and $c:=\sum_{i=1}^n x_i$. If $0\le c\le 1$, then
    $$
        \Pr\bigg[\bigvee_{i=1}^n A_i\bigg]\ge\sum_{k=1}^t(-1)^{k-1}\frac{c^k}{k!}.
    $$
\end{lemma}

To prove this, we make use of the Bonferroni inequalities, which generalize the inclusion-exclusion principle.
\begin{fact}[Bonferroni inequalities]\label{lem:bonferroni}
    Let $A_1,\ldots,A_n$ be $n$ random events. Then for any positive even integer $t$, we have
    \begin{equation}
        \Pr\bigg[\bigvee_{i=1}^nA_i\bigg]\ge\sum_{k=1}^t(-1)^{k-1}\sum_{S\sse[n],|S|=k}\Pr\bigg[\bigwedge_{i\in S}A_i\bigg]\label{eq:bonferroni}
    \end{equation}
\end{fact}

Still, let $A_1,\ldots,A_n$ be $n$ random events, and let $x_i:=\Pr[A_i]$. Suppose the events $\{A_i\}_{i=1}^n$ are $t$-wise independent for some positive even integer $t$, then \cref{lem:bonferroni} already shows that
\begin{equation}
    \Pr\bigg[\bigvee_{i=1}^nA_i\bigg]\ge\sum_{k=1}^t(-1)^{k-1}\sum_{S\sse[n],|S|=k}\ \prod_{i\in S}x_i \label{eq:twise-ind}
\end{equation}
Let $c:=\sum_{i=1}^nx_i$. We ask the following question: for some fixed $c$, what is the minimum value of RHS of \cref{eq:twise-ind}? In fact, for $c\le 1$, we are able to show that the minimum is achieved when $x_1=\cdots=x_n=\frac cn$.

\begin{claim}\label{lem:twise-opt}
    Let $x_1,\ldots,x_n$ be $n$ variables, let $c\in[0,1]$ be some fixed number, and let $t$ be some positive even integer. Then for the optimization problem
    \begin{align*}
        \min_{\bx}\quad&f_t(\bx)=\sum_{k=1}^t(-1)^{k-1}\sum_{S\sse[n],|S|=k}\ \prod_{i\in S}x_i\\
        \text{s.t.}\quad&\bx\ge\mathbf0\\
        &\sum_{i=1}^n x_i=c,
    \end{align*}
    The unique minimum value is $\sum_{k=1}^t(-1)^{k-1}\binom nk\frac{c^k}{n^k}$, obtained when $x_1=\cdots=x_n=\frac cn$.
\end{claim}
\begin{proof}
    Suppose for contradiction that for some $p$ and $q$ we have $x_p\ne x_q$ in the optimal solution $\bx_{opt}$ (which must exist since $f_t$ is continuous and the feasible set is compact). Then define $\bx_{opt}'$ as the solution where we substitute $(x_p,x_q)$ by $(\frac{x_p+x_q}2,\frac{x_p+x_q}2)$ in $\bx_{opt}$. We will argue that $f_t(\bx_{opt}')<f_t(\bx_{opt})$.

    One can verify that $f_t(\bx_{opt})$ and $f_t(\bx_{opt}')$ can be represented as
    \begin{align*}
        f_t(\bx_{opt})&=1-(1-x_p)(1-x_q)\bigg(\sum_{k=0}^{t-2}(-1)^k\sum_{S\sse[n]\setminus\{p,q\},|S|=k}\ \prod_{i\in S}x_i\bigg)\\
        f_t(\bx_{opt}')&=1-\Big(1-\frac{x_p+x_q}2\Big)\Big(1-\frac{x_p+x_q}2\Big)\bigg(\sum_{k=0}^{t-2}(-1)^k\sum_{S\sse[n]\setminus\{p,q\},|S|=k}\ \prod_{i\in S}x_i\bigg).
    \end{align*}
    The sum in the second term is positive because
    \begin{align*}
        \sum_{k=0}^{t-2}(-1)^k\sum_{S\sse[n]\setminus\{p,q\},|S|=k}\ \prod_{i\in S}x_i &\ge\sum_{k=0}^{\frac t2-2}\bigg(\sum_{S\sse[n]\setminus\{p,q\},|S|=2k}\ \prod_{i\in S}x_i\bigg)-\bigg(\sum_{S\sse[n]\setminus\{p,q\},|S|=2k+1}\ \prod_{i\in S}x_i\bigg)\\
        &\ge\sum_{k=0}^{\frac t2-2}\bigg(\sum_{S\sse[n]\setminus\{p,q\},|S|=2k}\ \prod_{i\in S}x_i\bigg)\bigg(1-\sum_{j\in[n]\setminus\{p,q\}}x_j\bigg)\\
        &\ge1-\sum_{i\in[n]\setminus\{p,q\}} x_i\tag{dropping all terms except $k=0$}\\
        &>1-x\tag{$x_p+x_q>0$ since $x_p\ne x_q$ and $x_p,x_q\ge0$}\\
        &\ge0.
    \end{align*}
    Also for $x_p\ne x_q$, we have $(1-x_p)(1-x_q)<(1-\frac{x_p+x_q}2)^2$. Therefore we can conclude that $f_t(\bx_{opt}')<f_t(\bx_{opt})$, which contradicts the assumption that $\bx_{opt}$ is a optimal solution.
\end{proof}

Combining \cref{lem:bonferroni} and \cref{lem:twise-opt}, we have a lower bound for the probability of union of events:
$$
    \Pr\bigg[\bigvee_{i=1}^nA_i\bigg]\ge\sum_{k=1}^t(-1)^{k-1}\frac{c^k}{n^k}\binom nk
$$
Expanding terms shows that this lower bound is $\sum_{k=1}^t(-1)^{k-1}\frac{c^k}{k!} - O(\frac 1n)$ for constant $t$. In fact, a more careful analysis shows it is actually no less than $\sum_{k=1}^t(-1)^{k-1}\frac{c^k}{k!}$.

\begin{claim}\label{lem:exp-series}
    For real number $c\in[0,1]$ and positive even integer $t$, we have
    $$
        \sum_{k=1}^t(-1)^{k-1}\frac{c^k}{n^k}\binom nk\ge\sum_{k=1}^t(-1)^{k-1}\frac{c^k}{k!}.
    $$
\end{claim}
\begin{proof}
    We can regroup the terms in the LHS as
    \begin{align*}
        \sum_{k=1}^t(-1)^{k-1}\frac{c^k}{n^k}\binom nk &=c-\frac{c^t}{n^t}\binom nt-\sum_{k=1}^{\frac t2-1}\frac{c^{2k}}{n^{2k}}\binom{n}{2k}\Big(1-\frac{c(n-2k)}{n(2k+1)}\Big).
        \intertext{Similarly for the RHS,}
        \sum_{k=1}^t(-1)^{k-1}\frac{c^k}{k!} &=c-\frac{c^t}{t!}-\sum_{k=1}^{\frac t2-1}\frac{c^{2k}}{(2k)!}\Big(1-\frac{c}{2k+1}\Big).
    \end{align*}
    To compare LHS and RHS, we compare each of the three summands separately. For the first part, obviously $c\ge c$. For the second part, one can verify that
    $-\frac{c^t}{n^t}\binom nt\ge-\frac{c^t}{t!}.$
    Each term in the sum of the third part is positive, so to compare LHS and RHS, we take the quotient between them:
    \begin{align*}
        \frac{\frac{c^{2k}}{n^{2k}}\binom{n}{2k}(1-\frac{c(n-2k)}{n(2k+1)})}{\frac{c^{2k}}{(2k)!}(1-\frac{c}{2k+1})}
        &=\frac{\prod_{i=0}^{2k-1}(n-i)}{n^{2k}}\Big(1+\frac{2kc}{n(2k+1-c)}\Big) \le\Big(1-\frac1n\Big)\Big(1+\frac{2kc}{n(2k+1-c)}\Big)\\
        &\le\Big(1-\frac1n\Big)\Big(1+\frac1n\Big)\tag{$c\le1$ implies $\frac{2kc}{2k+1-c}\le1$}\\
        &\le1.
    \end{align*}
    Hence we conclude that $\text{LHS}\ge\text{RHS}$.
\end{proof}

Combining \cref{lem:bonferroni}, \cref{lem:twise-opt} and \cref{lem:exp-series}, we may conclude \cref{cor:lb-union} as a corollary. Finally, we can combine this with the following result of \textcite{DBLP:journals/siamcomp/ChekuriVZ14} to prove our lower bound.

\begin{theorem}[Section 4.2 from \cite{DBLP:journals/siamcomp/ChekuriVZ14}]
    Let $M=(E,\cI)$ be a matroid, and $\cD$ a distribution over subsets of $N$. Let $R\sse E, R\sim \cD$ denote the random active set of elements, and $x_i=\Pr[i\in R]$ denote the probability of $i$ being \emph{active}. Then the optimum CRS for $\cD$ is $c(M,\cD)$-balanced, where $c(M,\cD)$ is defined by
    $$
        c(M,\cD):=\min_{\by\ge0}\frac{\mathbb E_{R\sim \cD}[\rank_{\by}(R)]}{\sum_{i\in E}x_iy_i}.
    $$
    Here $\rank_\by(R)$ is defined as $\rank_\by(R):=\max_{I\sse R,I\in\cI}\sum_{i\in I}y_i$.
\end{theorem}

Our main goal in this section is to lower bound $c(M,\cD)$ where $M$ is a uniform matroid of rank 1, and $\cD$ is $t$-wise independent over the ground set. In particular, we will show the following lemma:

\begin{lemma}[$t$-wise Independent Prophet Inequality Lower Bound]\label{lemma-single-ratio}
    For any integer $t \geq 2$ and any $t$-wise independent distributions, there exists a $\phi(t)$-balanced single item CRS.
\end{lemma}
\begin{proof}
    Assume for now that $t$ is even. Let $p_i$ denote the probability that $y_i$ is the largest in $R$ (breaking ties towards the smallest index).
    When $M$ is a uniform matroid of rank 1, we can write
    $$
        c(M,\cD)=\min_{\by\ge0}\frac{\mathbb E_{R\sim\cD}[\max_{j\in R}y_j]}{\sum_{i\in E}x_iy_i} = \min_{\by\ge0}\frac{\sum_{i\in E}p_iy_i}{\sum_{i\in E}x_iy_i}.
    $$
    
    To bound this quantity, we renumber the items such that $y_1\ge y_2\ge\cdots\ge y_n$ and consider the following linear program:
    \begin{align*}
        \min\quad&\sum_{i=1}^np_iy_i\\
        \text{s.t.}\quad&y_1\ge\cdots\ge y_n\ge0\\
        &\sum_{i=1}^nx_iy_i=1.
    \end{align*}
    Notice that once the relative order of $y_i$ is fixed, $\bp$ is also fixed. Therefore $\bp$ in the above optimization problem can be treated as constants. To bound this LP, we use the sum by parts trick:
    $$
        \sum_{i=1}^nx_iy_i=\sum_{i=1}^n(y_i-y_{i+1})\bigg(\sum_{j=1}^ix_j\bigg),
    $$
    where we use the convention that $y_{n+1}=0$. Now if we define $Y_i:=y_i-y_{i+1}$, $P_i:=\sum_{j=1}^i p_j$, and $X_i:=\sum_{j=1}^i x_j$, our linear program can be rewritten as
    \begin{align*}
        \min\quad&\sum_{i=1}^nP_iY_i\\
        \text{s.t.}\quad&\mathbf Y\ge\mathbf0\\
        &\sum_{i=1}^n X_iY_i=1.
    \end{align*}
    The minimum of this program is $\min_i\frac{P_i}{X_i}=\min_i\frac{\sum_{j=1}^ip_j}{\sum_{j=1}^ix_j}$. Now consider the meaning of $\sum_{j=1}^ip_j$:
    \begin{align*}
        \sum_{j=1}^ip_j&=\Pr\Big[\argmax_{j\in R}y_j\in[i]\Big] =\Pr[R\cap[i]\ne\emptyset] =\Pr\bigg[\bigvee_{j=1}^i \{j\in R \}\bigg].
    \end{align*}
    Let $c_i:=\sum_{j=1}^i x_j$. Then by \cref{cor:lb-union}, $\Pr[\bigvee_{j=1}^i \{j\in R\}]\ge\sum_{k=1}^t(-1)^{k-1}\frac{c_i^k}{k!}$. Therefore the minimum of the linear program is at least
    $$
        \frac{\sum_{k=1}^t(-1)^{k-1}\frac{c_i^k}{k!}}{c_i}=\sum_{k=1}^t(-1)^{k-1}\frac{c_i^{k-1}}{k!}
    $$
    for some $i$.
    We claim that $\sum_{k=1}^t(-1)^{k-1}\frac{c_i^{k-1}}{k!}\ge\phi(t)=\sum_{k=1}^t(-1)^{k-1}\frac{1}{k!}$. To see this, define $f_t(x)$ as
    $$
        f_t(x):=\sum_{k=1}^t(-1)^{k-1}\frac{x^{k-1}}{k!}.
    $$
    Our claim can be rewritten as $f_t(c_i)\ge f_t(1)$, so it suffices to show that $f_t'(x)\le0$ in $[0,1]$. We have
    \begin{align*}
        f_t'(x)&=\sum_{k=2}^t(-1)^{k-1}\frac{(k-1)x^{k-2}}{k!} =-\frac{(t-1)x^{t-2}}{t!}-\sum_{k=1}^{\frac t2-1}\Big(\frac{(2k-1)x^{2k-2}}{(2k)!}-\frac{(2k)x^{2k-1}}{(2k+1)!}\Big)\\
        &=-\frac{(t-1)x^{t-2}}{t!}-\sum_{k=1}^{\frac t2-1}\frac{(2k-1)x^{2k-2}}{(2k)!}\Big(1-\frac{2k}{4k^2-1}x\Big) \le 0.
    \end{align*}
    Note that we assumed in the last chain of equations that $t$ is even. When $t$ is odd, the optimal CRS is no worse than the best possible $(t-1)$-wise independent CRS, and hence the bound is as claimed. This concludes the proof.
\end{proof}
\subsection{A $(\sqrt2-1)$-Balanced PI-OCRS for a Single Item}
\label{sec:sqrt2}

Caragiannis et al. \cite{DBLP:conf/wine/CaragiannisGLW21} give a $(\sqrt2-1)$-prophet inequality for the single-item case; we show how the same ideas give a $(\sqrt2-1)$-balanced PI-OCRS as well, when the adversary is oblivious. (In \Cref{sec:app-uniform-simpleOCRS} we show a uniform threshold algorithm that works against the almighty adversary as well.) Indeed, when item $i$ arrives and none of the items have been selected, we flip a coin and select it with probability $q_i$. For simplicity, let $X_i=\mathbbm1[i\in R]$, and let Bernoulli random variable $Y_i$ denote the event that the coin flip is heads. Let $I$ denote the final set we select. Then for item $i$, we can bound $\Pr[i\in I\mid i\in R]$ by
\begin{align*}
    \Pr[i\in I\mid i\in R]&=\Pr[Y_i=1\mid X_i=1]\cdot\Pr[X_jY_j=0\ \forall j<i\mid X_i=1]\\
    &=q_i\bigg(1-\Pr\bigg[\bigvee_{j<i}X_jY_j=1\ \bigg|\  X_i=1\bigg]\bigg)\\
    &\ge q_i\bigg(1-\sum_{j<i}\Pr[X_jY_j=1\mid X_i=1]\bigg) \tag{union bound}\\
    &=q_i\bigg(1-\sum_{j<i}\Pr[X_jY_j=1]\bigg)\tag{pairwise independence of $X_i$}\\
    &=q_i\bigg(1-\sum_{j<i}x_jq_j\bigg) \tag{$Y_i$ independent}
\end{align*}
Thus we want to choose $q_i$ such that $c=\min_i q_i(1-\sum_{j<i}x_jq_j)$ is as large as possible. A simple idea is to set $q_i=1/2$ for all $i$, which guarantees that $c\ge \nicefrac{1}{4}$. In fact, we choose $q_i$ such that $c\ge\sqrt2-1$, and \Cref{lem:singleitem_optimal} gives an example  illustrating that $\sqrt2-1$ is asymptotically the best possible for this approach.

\begin{lemma}
    Let $\bx\in\mathbb R_+^n$ be a real vector such that $\sum_i x_i\le1$. Then there exists $\mathbf q$ such that $q_i(1-\sum_{j<i}x_jq_j)\ge\sqrt2-1$ holds for all $i$. Moreover, $q_i$ can be calculated online: we can calculate such $q_i$ based only on $x_1,\cdots,x_i$.
\end{lemma}
\begin{proof}
    We set $q_i=f(\sum_{j<i}x_j)$ where $f(t)=1/\sqrt{(3+2\sqrt2)-(2+2\sqrt2)t}$ to show that $q_i(1-\sum_{j<i}x_jq_j)\ge\sqrt2-1$, we need the following claim.
    \begin{claim}
        Let $f:[0,1]\to\mathbb R$ be a continuous non-decreasing function, and $\bx\in R_+^n$ a non-negative vector such that $\sum_i x_i\le1$. Let $S_i$ denote $\sum_{j=1}^ix_i$. Then for any $i$,
        $$
            \sum_{j<i}x_j \cdot f(S_{j-1})\le\int_0^{S_{i-1}}f(t)\:\text{d}t.
        $$
    \end{claim}
    \begin{proof}
        Since $f$ is non-decreasing, we have
        $$
            x_j\cdot f(S_{j-1})\le\int_{S_{j-1}}^{S_j}f(t)\:\text{d}t.
        $$
        Then summing over $j$ proves the claim.
    \end{proof}
    Let $S_i$ denote the prefix sum of $\bx$, i.e. $S_i:=\sum_{j=1}^i x_j$. Then we can bound $q_i(1-\sum_{j<i}x_jq_j)$ by
    \begin{align*}
        q_i\bigg(1-\sum_{j<i}x_jq_j\bigg) &=f(S_{i-1})\bigg(1-\sum_{j<i}x_jf(S_{j-1})\bigg) \ge f(S_{i-1})\bigg(1-\int_0^{S_{i-1}}f(t)\text dt\bigg)\\
        &=\frac{1}{\sqrt{(3+2\sqrt2)-(2+2\sqrt2)S_{i-1}}}\cdot(\sqrt2-1)\sqrt{(3+2\sqrt2)-(2+2\sqrt2)S_{i-1}}\\
        &=\sqrt2-1. \qedhere
    \end{align*}
\end{proof}

We now show that even for uniform $\bx=(1/n,\cdots,1/n)$,  $(\sqrt2-1)$ is asymptotically the best no matter what choice of $\mathbf q$ we use.
\begin{lemma} \label{lem:singleitem_optimal}
    For every $\mathbf q\in[0,1]^n$, we have
    $$
        \min_iq_i\left(1-\frac1n\sum_{j<i}q_j\right)\le\sqrt2-1+O\left(\frac{1}{n}\right).
    $$
\end{lemma}
\begin{proof}
    Let $r_i=q_i(1-\frac1n\sum_{j<i}q_j)$. First, for any non-negative vector $\mathbf p\in\mathbb R_+^n$, we have
    $$
        \min_i r_i\le\frac{\mathbf r\cdot\mathbf p}{\|\mathbf p\|_1}.
    $$
    In particular, the $\mathbf p$ we choose is $(1,\cdots,1,\alpha n)$, i.e. first $(n-1)$ entries are $1$, while the last entry is $\alpha n$, where $\alpha$ is some constant whose value is to be determined. Then we compute $\mathbf r\cdot\mathbf p$:
    \begin{align*}
        \mathbf r\cdot\mathbf p&=\sum_{i=1}^{n-1}q_i\bigg(1-\frac1n\sum_{j<i}q_j\bigg)+\alpha nq_n\bigg(1-\frac1n\sum_{j<n}q_i\bigg)\\
        &=(1-\alpha q_n)\sum_{i=1}^{n-1}q_i-\frac1n\sum_{1\le j<i\le n-1}q_iq_j+\alpha nq_n\\
        &=(1-\alpha q_n)\sum_{i=1}^{n-1}q_i-\frac1{2n}\left(\sum_{i=1}^{n-1} q_i\right)^2+\frac1{2n}\sum_{i=1}^{n-1}q_i^2+\alpha nq_n\\
        &\le\frac n2(1-\alpha q_n)^2+\alpha nq_n+\frac12.
    \end{align*}
    Meanwhile, $\|\mathbf p\|_1=n-1+\alpha n$. Therefore we have
    \begin{align*}
        \min_i r_i&\le\frac{\mathbf r\cdot\mathbf p}{\|\mathbf p\|_1} \le\frac{\frac12(1-\alpha q_n)^2+\alpha q_n+\frac1{2n}}{\frac{n-1}n+\alpha} =\frac{1+\alpha^2q_n^2}{2+2\alpha}+O\left(\frac1n\right) \le\frac{1+\alpha^2}{2+2\alpha}+O\left(\frac1n\right).
    \end{align*}
    The expression $(1+\alpha^2)/(2+2\alpha)$ has a minimum value $\sqrt2-1$, achieved when $\alpha=\sqrt2-1$. Therefore we conclude that $\min_i r_i\le \sqrt2-1+O(\frac1n)$.
\end{proof}
The above analysis requires that the arrival order of the items remain the same regardless of $R$, and so this PI-OCRS may not be be robust against the almighty adversary. In \Cref{sec:single-sample} we show a simple $\nicefrac{1}{4}$-selectable PI-OCRS for all uniform matroids, including the single item case.

\subsection{An Upper Bound for Multiple-Threshold Algorithms}
\label{sec:imposs-single}

Here we show an example where no multiple-threshold algorithm can give an approximation ratio better than $2\sqrt5-4 \approx 0.472$. This result shows that we need new ideas to match the $\nicefrac12$-prophet inequality for the fully independent case. 

The instance has $n+2$ items. Item $1$ has a deterministic value of $t$, where $0<t<1$ is some constant to be determined; items $2,\cdots, n+1$ each has value $1$ with probability $\frac1n$ and $0$ otherwise; item $n+2$ has value $rn$ with probability $\frac1n$ and $0$ otherwise, where $r>0$ is some constant to be determined. The joint distribution is constructed as follows:
\begin{itemize}
    \item With probability $\frac{n+1}{2n}$, sample $2$ items from $2,\cdots,(n+2)$ and make them non-zero (i.e. for items $2,\cdots, n+1$, set value as $1$; for item $n+2$, set value as $rn$). The remaining items have value $0$.
    \item With probability $\frac{n-1}{2n}$, items $2,\cdots, n+2$ all have value zero.
\end{itemize}
It's easy to verify that this distribution is indeed pairwise independent, and the expected value of the largest item is
$$
    \text{OPT}=\frac{n-1}{2n}\cdot t+\frac1n\cdot rn+\left(\frac{n+1}{2n}-\frac1n\right)\cdot 1=r+\frac12(1+t)+O\left(\frac1n\right)
$$
Now for any multiple-threshold algorithm, its behaviour on this example can be described as:
\begin{itemize}
    \item On seeing item $i$, if $i$ is non-zero, toss an independent coin to take it with probability $q_i$; otherwise abandon this item. Here $q_i$ is only determined by the ordering of items, and does not depend on values of item $1\cdots i-1$.
\end{itemize}
Here we may as well assume $q_{n+2}=1$. And we can calculate the expected reward for each item:
\begin{itemize}
    \item For item $1$, the expected reward is $q_i\cdot t$
    \item For item $n+2$, the expected reward is
    $$
        \frac1n\cdot rn\cdot(1-q_1)\cdot\left(\sum_{i=2}^{n+1}\frac1n(1-q_i)\right)=r(1-q_1)\left(1-\frac1n\sum_{i=2}^{n+1}q_i\right)
    $$
    \item For items $2\cdots n+1$, the expected reward is
    \begin{align*}
        &\sum_{i=2}^{n+1}\frac1n\cdot q_i\cdot(1-q_1)\cdot\left(\sum_{j=2}^{i-1}\frac1n(1-q_j)+\sum_{j=i+1}^{n+2}\frac1n\right)\\
        &=(1-q_1)\left(\frac1n\sum_{i=2}^{n+1}q_i-\frac1{n^2}\sum_{2\le j<i\le n+1}q_iq_j\right)\\
        &=(1-q_1)\left(\frac1n\sum_{i=2}^{n+1}q_i-\frac1{2n^2}\left(\sum_{i=2}^{n+1}q_i\right)^2+\frac1{2n^2}\sum_{i=2}^{n+1}q_i^2\right)\\
        &=(1-q_1)\left(\frac1n\sum_{i=2}^{n+1}q_i-\frac1{2n^2}\left(\sum_{i=2}^{n+1}q_i\right)^2\right)+O\left(\frac1n\right).
    \end{align*}
\end{itemize}
Therefore the total expected reward for the algorithm is
\begin{align*}
    \text{ALG}&=q_1t+(1-q_1)\left(r+(1-r)\frac{\sum_{i=2}^{n+1}q_i}{n}-\frac12\left(\frac{\sum_{i=2}^{n+1}q_i}{n}\right)^2\right)+O\left(\frac1n\right)\\
    &\le q_1t+(1-q_1)\left(r+\frac12(1-r)^2\right)+O\left(\frac1n\right)\\
    &=q_1t+(1-q_1)\frac12(1+r^2)+O\left(\frac1n\right)\\
    &\le\max\left(t, \: \frac{1}{2}(1+r^2)\right)+O\left(\frac1n\right).
\end{align*}
Putting these expressions together, the approximation ratio for our algorithm is at most
$$
    \text{APX}=\frac{\text{ALG}}{\text{OPT}}\le\frac{\max(t,\frac12(1+r^2))}{r+\frac12(t+1)}+O\left(\frac1n\right).
$$
The minimal value is achieved when $t=\frac12(1+r^2)$, giving us
$$
    \text{APX}\le\frac{\frac12(1+r^2)}{\frac14r^2+r+\frac34}+O\left(\frac1n\right).
$$
The expression on the right is minimized when $r=\frac12(\sqrt5-1)$, giving us a minimum value of $2\sqrt5-4+O(\frac1n)$, as claimed.
\subsection{Upper Bound for Single Item PI-OCRS Against Almighty Adversary}

\label{app:almighty-ub}

In this section, we show  that there exists no single item $(\nicefrac14+\eps)$-selectable PI-OCRS. Our distribution $\cD$ for the random active set $R$ is defined as follows over the ground set $[n]$:
\begin{itemize}
    \item With probability $\frac1n$, uniformly sample one item $i$ from $[n]$ and set $R=\{i\}$.
    \item With probability $\frac{n-1}{2n}$, uniformly sample two items $i\ne j$ from $[n]$ and set $R=\{i,j\}$.
\end{itemize}
We define $x_i:=\Pr[i\in R]$. One can verify that for any item $i\ne j$, $x_i=x_j=\frac1n$, and $\Pr[i\in R,j\in R]=\frac1{n^2}=x_ix_j$. Therefore this distribution is indeed pair-wise independent, and $\bx$ is in the matroid polytope.

Let $S$ denote the string of random bits used by the algorithm. Now recall the definition of almighty adversary: First, $R$ is sampled according to $\cD$, and $S$ is sampled by the algorithm. The adversary can see $R$ and $S$. Then it presents the items to the algorithm in adversarial order. Noticed that after $S$ is fixed, the strategy of the algorithm is deterministic. We say item \emph{$i$ dominates $j$ for $S$} if, when $R=\{i,j\}$ and the random string is $S$, the algorithm always picks $i$ instead of $j$ regardless of the ordering chosen by the adversary. We use $i>_S j$ for short. We show some properties regarding the domination relationship.

\begin{lemma}\label{lem:dominate}
    For any algorithm and any fixed string $S$, for any item $i \in [n]$, if there exists $j \in [n]$ such that $i>_S j$, then there exists no $k$ such that $k>_S i$.
\end{lemma}
\begin{proof}
    If such $k$ exists, the adversary can order the items such that $i,j,k$ appear in the relative order of $\cdots i\cdots k\cdots j\cdots$. On the arrival of $i$, the algorithm cannot distinguish whether $R=\{i,j\}$ or $R=\{i,k\}$ (or $R=\{i,l\}$ for some other $l$). By definition, $i>_S j$ means the algorithm must select $i$, while $k>_S i$ means the algorithm must discard $i$, a contradiction.
\end{proof}

Let $E_i$ denote the random event (with respect to the distributions of $S$ and $R$) that $i \in R$ and the algorithm selects $i$ no matter the order chosen by the adversary. Then
\begin{align}
\Pr[E_i]&\le\frac1{n^2}+\frac1{n^2}\mathbb E_S[|\{j \ : \ i>_S j\}|], \notag \\
\intertext{and hence}
    \sum_{i=1}^n\Pr[E_i]&\le\frac1n+\frac1{n^2}\mathbb E_S[|\{(i,j) \ : \ i>_S j \}|].\label{eq:sumei}
\end{align}
Denote $A_S:=\{i\in [n] \ : \ \exists j\in [n] \text{ s.t. } i>_S j\}$ and $B_S:=\{j\in [n] \ : \ \exists i\in [n] \text{ s.t. } i>_S j\}$. By \Cref{lem:dominate}, $A_S\cap B_S=\emptyset$. Therefore,
$$
    \mathbb E_S[|\{(i,j) \ : \ i>_S j \}|]\le\mathbb E[|A_S||B_S|]\le\frac{n^2}4
$$
Substituting back into \eqref{eq:sumei}, we have
$$
    \sum_{i=1}^n\Pr[E_i]\le\frac14+\frac1n.
$$
By averaging, there must exists some $i_0$ such that $\Pr[E_{i_0}]\le\frac{4n+1}{n^2}$. When $E_{i_0}$ does not occur, the adversary can always alter the order of items to make the algorithm not select $i_0$. For such an adversary, we have
$$
    \Pr[i_0\text{ is selected} \mid i_0\in R]\le\frac{\Pr[E_{i_0}]}{\Pr[i\in R]}\le\frac1n+\frac14.
$$
which finishes our proof.
\section{Details for Uniform Matroids}

\subsection{A $(b,1-b)$-Selectable PI-OCRS for Uniform Matroid}
\label{sec:app-uniform-simpleOCRS}

Here we show a $(b,1-b)$-selectable PI-OCRS for uniform matroids. Let $M=(E,\cI)$ be a $k$-uniform matroid, where $E$ is identified as $[n]$. Let $\cD$ be some distribution PI-consistent with $\bx\in b\cP_M$, and let $R$ be sampled according to $\cD$. Then we simply set the feasible set family $\cF=\cI$. We then bound the selectability of this PI-OCRS using Markov's inequality:
\begin{align*}
    &\Pr[I\cup\{i\}\in\cF\quad\forall I\in\cF,I\sse R\mid i\in R]\\
    &=\Pr[|R\setminus\{i\}| \le k-1\mid i\in R]\\
    &=1-\Pr[|R\setminus\{i\}|\ge k\mid i\in R]\\
    &\geq 1-\frac{\mathbb E[|R\setminus\{i\}|\mid i\in R]}{k}\tag{Markov's inequality}\\
    &=1-\frac{\sum_{i'\in E\setminus\{i\}}\Pr[i'\in R\mid i\in R]}{k}\\
    &=1-\frac{\sum_{i'\in E\setminus\{i\}}\Pr[i'\in R]}{k}\tag{pairwise independence of events $i\in R$}\\
    &\ge1-b\tag{$\bx(E)\le bk$}
\end{align*}
Therefore $\cF$ defines a $(b,1-b)$-selectable PI-OCRS.

\subsection{A $(1-O(k^{-\nicefrac13})$-Balanced PI-CRS for Uniform Matroid}
\label{sec:uniform-CRS}

First, assume that we are given any
distribution $\cD$ PI-consistent with some $\bx\in(1-\eps)\cP_M$.
Suppose that $\bx(E) = (1-\delta)k$ for some
    $\delta\ge\eps$. By \Cref{clm:averaging2} and an
    averaging argument, there exists $i_n$ such that
    \[ \Pr[\abs{R}\ge k \mid i_n \in
      R]=\frac{\Pr[\abs{R}\ge k, i_n\in R]}{\Pr[i_n\in R]}\le\frac{\frac{1-\delta^2}{\delta^2}}{(1-\delta)k}=\frac{1+\delta}{\delta^2k}\le\frac{1+\eps}{\eps^2k}. \]
After we select $i_n$ as above, observe that $M\setminus\{i_n\}$ is also a $k$-uniform matroid, and $\bx(E\setminus\{i_n\})\le (1-\eps)k$. Therefore using the same argument as above, we can choose $i_{n-1}$ such that $\Pr[\abs{R\setminus\{i_n\}}\ge k\mid i_{n-1}\in R]\le\frac{1+\eps}{\eps^2k}$. Similarly, we can choose $i_{n-2}$ such that $\Pr[\abs{R\setminus\{i_n,i_{n-1}\}}\ge k\mid i_{n-2}\in R]\le\frac{1+\eps}{\eps^2k}$. We repeat this procedure, and obtain an ordering of items $\{i_j\}_{j=1}^n$ such that for any $j$,
$$
    \Pr[\abs{R\cap\{i_{j'}\mid j'\le j\}}\ge k\mid i_{j}\in R]\le\frac{1+\eps}{\eps^2k}.
$$

Let $I$ denote the set of items we select, initially empty. Now we consider the item one by one from $i_1$ to $i_n$ and run a ``greedy'' procedure: we add $i_j \in R$ to $I$ whenever $|I\cup\{i_j\}|\le k$, and finally return $\pi(R)=I$. Since $\Pr[\abs{R\cap\{i_j'\mid j'\le j\}}\ge k\mid i_j\in R]\le\frac{1+\eps}{\eps^2k}$, we have $\Pr[i\in \pi(R)\mid i\in R]\ge1-\frac{1+\eps}{\eps^2k}$. In conclusion, the above procedure
gives an $(1-\eps,1-(1+\eps)/(\eps^2 k))$-balanced PI-CRS.  If we set
$\eps=k^{-\nicefrac13}$, by \Cref{clm:scaleCRS} we have an $(1-O(k^{-\nicefrac13}))$-balanced PI-CRS.

\section{Details for Laminar Matroids}
\label{sec:app-laminar}

Let $b$ be some constant close to $1$. Each constraint $(A,c'(A))$ is equivalent to a $c'(A)$-uniform matroid over $A$, therefore by \Cref{sec:uniform}, we have a $(1-b,b)$-selectable PI-OCRS as well as a $(1-b,1-(\frac{4}{27}b^3c'(A))^{-\nicefrac12})$-selectable PI-OCRS. We set some threshold $t$ to use the first PI-OCRS when $c'(A)< 2^t$, and the second when $c'(A)\ge2^t$. Let $\cF_{A,c'(A)}$ denote the set family for the respective greedy PI-OCRS. We define the \emph{extension} of $\cF_{A,c'(A)}$ to $E$ as $\cF^E_{A,c'(A)}:=\{I\sse E\mid I\cap A\in\cF_{A,c'(A)}\}$. Finally, we set $\cF_{\pi,\cD}:=\bigcap_{A\in\cA'}\cF^E_{A,c'(A)}$, and our greedy PI-OCRS is defined by $\cF_{\pi,\cD}$.

First, we claim that $\cF_{\pi,\cD}\sse\cI'$, where $\cI'$ denotes the independent sets of $M'$. This is true because all sets in $\cF^E_{A,c'(A)}$ satisfies the constraint $(A,c'(A))$, and therefore $\cF_{\pi,\cD}$ satisfies all the constraints as the intersection of all $\cF^E_{A,c'(A)}$. To bound the selectability of the greedy PI-OCRS defined by $\cF_{\pi,\cD}$, we first need the following observation:
\begin{align}
    &\Pr[I\cup\{i\}\in\cF^E_{A,c'(A)}\quad\forall I\in\cF^E_{A,c'(A)},I\sse R\mid i\in R]\notag\\
    &=\begin{cases}
        1&,i\notin A\\
        \Pr[I\cup\{i\}\in\cF_{A,c'(A)}\quad\forall I\in\cF_{A,c'(A)},I\sse R\cap A\mid i\in R\cap A]&, i\in A
    \end{cases}\label{eq:laminargreedy}
\end{align}
Then we have
\begin{align*}
    &\Pr[I\cup\{i\}\in\cF_{\pi,\cD}\quad\forall I\in\cF_{\pi,\cD},I\sse R\mid i\in R]\\
    &=\Pr\bigg[I\cup\{i\}\in\bigcap_{A\in\cA'}\cF^E_{A,c'(A)}\quad\forall I\in\bigcap_{A\in\cA'}\cF^E_{A,c'(A)},I\sse R\mid i\in R\bigg]\\
    &\ge1-\sum_{A\in\cA'}\bigg(1-\Pr\bigg[I\cup\{i\}\in\cF^E_{A,c'(A)}\quad\forall I\in\bigcap_{B\in\cA'}\cF^E_{B,c'(B)},I\sse R\mid i\in R\bigg]\bigg)\\
    &\ge1-\sum_{A\in\cA'}(1-\Pr[I\cup\{i\}\in\cF^E_{A,c'(A)}\quad\forall I\in\cF^E_{A,c'(A)},I\sse R\mid i\in R])\\
    &=1-\sum_{A\ni i}(1-\Pr[I\cup\{i\}\in\cF_{A,c'(A)}\quad\forall I\in\cF_{A,c'(A)},I\sse R\cap A\mid i\in R\cap A])\tag{\Cref{eq:laminargreedy}}
\end{align*}
Since $\cA'$ is a laminar set family, for $A,B\in\cA',A\ne B,A\cap B\ni i$, either $A\subsetneq B$ or $B\subsetneq A$. Therefore by the strict monotonicity of $c'$, $c'(A)$ for $A\ni i$ form a geometric series. Recall that we use the $(1-b,b)$-selectable PI-OCRS for $c'(A)<2^t$, and $(1-b,1-(\frac{4}{27}b^3c'(A))^{-1/2})$-selectable PI-OCRS for $c'(A)\ge 2^t$. Therefore we have
\begin{align*}
    &1-\sum_{A\ni i}\left(1-\Pr\left[I\cup\{i\}\in\cF_{A,c'(A)}\quad\forall I\in\cF_{A,c'(A)},I\sse R\cap A\mid i\in R\cap A\right]\right)\\
    &\geq 1 - t \:(1- b) - \sum_{i \geq t} \left( \frac{4}{27} \: b^3 \: 2^i \right)^{-\nicefrac12} \\
    &= 1 - t \: (1-b) -  b^{-\nicefrac32} \: 2^{-\nicefrac t2} \: \frac{3 \sqrt{3}}{2 - \sqrt{2}}.
\end{align*}
Taking $t=13$ and $b=\nf{24}{25}$ gives a $(\nf1{25},\nf1{2.661})$-selectable PI-OCRS. Therefore by \Cref{clm:scaleOCRS}, we have a $\nf1{67}$-selectable PI-OCRS for $M'$, which by the
observation above gives a $(\nf12, \nf1{67})$-selectable, and again by \Cref{clm:scaleOCRS}, a $\nf1{134}$-selectable PI-OCRS for $M$.

\section{Limitations of Our Approach for Graphic Matroids}
\label{sec:limit-graphic}

In \Cref{sec:graphic}, we showed that for any graphic matroid $M=(E,\cI)$ and $R$ sampled according to any distribution $\cD$ PI-consistent with some $\bx\in b\cP_M$, there exists some edge $e_0\in E$ such that $\Pr[e_0\in\spn(R\setminus\{e_0\})\mid e_0\in R]\le 2b$. Here we give an example showing that the $2b$ here is asymptotically tight. Consider the complete graph $K_n$ where $n$ is odd. Our construction of the distribution $\cD$ is quite simple: with probability $\frac{n(n-1)}{(n+3)(n-2)}$, uniformly sample a cycle of length $\frac{n+3}2$ from $K_n$. Otherwise take no edge at all.

First, we want to verify that this is indeed a pairwise-independent distribution. Since all edges are equivalent, we have
\[
    \forall e\in E, \Pr[e\in R]=\frac{n(n-1)}{(n+3)(n-2)}\cdot\frac{(n+3)/2}{n(n-1)/2}=\frac1{n-2}.
\]
For pairs of edges $e_1$ and $e_2$, we have to consider two cases.
\begin{enumerate}
    \item $e_1$ and $e_2$ are \emph{adjacent}, meaning that they share one common vertex. Notice that a cycle of length $\frac{n+3}2$ contains $\frac{n+3}2$ adjacent edge pairs. By taking the expectation, we can see that
    \begin{align*}
        \sum_{e_i,e_j\text{ adjacent}}\Pr[e_i\in R\land e_j\in R] &=\mathbb E[\#\text{ adjacent edge pairs in }R] =\frac{n(n-1)}{(n+3)(n-2)}\cdot\frac{n+3}2.
    \end{align*}
    The number of adjacent edge pairs in $K_n$ is $\binom n3\cdot3$, and by symmetry, they all have same probability of appearing in $R$. Therefore we have
    $$
        \Pr[e_1\in R\land e_2\in R]=\frac{n(n-1)}{2(n-2)}\cdot\frac1{3\binom n3}=\frac1{(n-2)^2}.
    $$
    \item $e_1$ and $e_2$ are not adjacent, i.e., they share no common vertices. A cycle of length $\frac{n+3}2$ contains $\binom{(n+3)/2}{2}-\frac{n+3}2=\frac{(n+3)(n-3)}8$ non-adjacent edge pairs. By taking the expectation, we have
    \begin{align*}
        \sum_{e_i,e_j\text{ not adjacent}}\Pr[e_i\in R\land e_j\in R]&=\mathbb E[\#\text{ non-adjacent edge pairs in }R]\\
        &=\frac{n(n-1)}{(n+3)(n-2)}\cdot\frac{(n+3)(n-3)}8.
    \end{align*}
    Similarly, the number of non-adjacent edge pairs in $K_n$ is $\binom n4\cdot3$, and by symmetry we have
    $$
        \Pr[e_1\in R\land e_2\in R]=\frac{n(n-1)(n-3)}{8(n-2)}\cdot\frac1{3\binom n4}=\frac1{(n-2)^2}
    $$
\end{enumerate}
So for both cases, the distribution is indeed pairwise-independent. (Here we may as well calculate the probability of edge pairs directly, but this approach is easier.) If we set $x_i=\frac{2}{n}$ for each edge $i$, then $\bx$ is in the matroid polytope, and we have $b=\frac{1/(n-2)}{2/n}=\frac{n}{2(n-2)}$. Notice that for any $e\in E$,
$$
    \Pr[e\in\spn(R\setminus\{e\})\mid e\in R]=1=\frac{2(n-2)}{n} \cdot b
$$
because if $e\in R$ then it is contained in a cycle. Hence $2b$ is asymptotically tight.

\end{document}